\documentclass[twoside,11pt]{article}

\usepackage{jmlr2e}
\usepackage{amsmath, amssymb, mathbbol}

\ShortHeadings{Results on Shannon Entropy Estimation in Infinite Alphabets}{Silva}
\firstpageno{1}

\begin{document}

\title{Shannon Entropy Estimation in $\infty$-Alphabets from Convergence Results} 

\author{\name Jorge F. Silva \email josilva@ing.uchile.cl \\
       \addr Department of Electrical Engineering\\
       Information and Decision System Group\\
       University of Chile\\
       Santiago, Chile, Off. 508, Ph. 56-2-9784090 
       }
\editor{}

\maketitle

\begin{abstract}
The problem of Shannon entropy estimation in countable infinite alphabets is addressed from the study and use of convergence results of the entropy functional,  which is known to be discontinuous  with respect to the total variation distance in $\infty$-alphabets. Sufficient conditions for the  convergence of the entropy are used,  including scenarios with both finitely and infinitely supported assumptions on the distributions. From this new perspective,  four plug-in histogram-based estimators are studied showing that convergence results are instrumental to derive new strong consistency and rate of convergences results.  Different scenarios and conditions are used on both the  estimators and the underlying distribution, considering for example  finite and unknown supported assumptions and summable tail bounded conditions. 
\end{abstract}


\section{Introduction}
The problem of Shannon entropy estimation has a long history in information theory, statistics and computer science \citep{beirlant_1997}. 
This problem belongs to the category of scalar functional estimation that has been richly studied in non-parametric statistics. Starting with the  finite size alphabet scenario,  the classical plug-in estimate (i.e., the empirical distribution evaluated on the functional) is well known to be consistent, minimax optimal and asymptotically efficient \citep[Secs. 8.7-8.9]{vaart_2000} when the number of samples $n$ goes to infinity.
In the finite alphabet setting, more recent research has been interested in looking at the called large alphabet regime,  meaning a non-asymptotic under-sampling regime where the numbers of samples $n$ is on the order of, or even smaller than, the size of the alphabet  denoted by $k$. In this 
context,  it has been shown that the classical plug-in estimator is sub-optimal as it suffers from severe bias \citep{wu_2016,jiao_2015}. 
For characterizing optimality in this high dimensional context,  a non-asymptotic minimax mean square error analysis under a finite $n$ and $k$ has been conducted by several authors \citep{paniski_2004,valiant_2011,valiant_2010,wu_2016,jiao_2015} considering the minimax risk $R^*(k,n)$\footnote{$R^*(k,n) = \inf_{\hat{H}(\cdot)} \sup_{\mu \in \mathcal{P}(k)} \mathbb{E}_{X_1,..X_n\sim \mu^n} \left\{ \left( \hat{H}(X_1,..,X_n) - H(\mu) \right)^2 \right\} $ where $\mathcal{P}(k)$ denotes the collection of probabilities on $[k] \equiv \left\{1,..,k \right\}$.}. \cite{paniski_2004} first showed that it was possible to construct an entropy estimator that uses a sub-linear sampling size to achieve minimax consistency when $k$ goes to infinity,  in the sense that there is a sequence $(n_k)=o(k)$ where $R^*(k,n_k) \longrightarrow 0$ as $k$ goes to infinity.  A 
set of results by \cite{valiant_2011,valiant_2010} show that the optimal scaling of the sampling size with respect to $k$ to achieve the aforementioned asymptotic consistency for entropy estimation is $O(k/\log(k))$.  A refined set of results for the complete characterization of $R^*(k,n)$,  the specific scaling of the sampling complexity, and  the achievability of the obtained minimax $L_2$ risk for the family $\left\{\mathcal{P}(k): k\geq 1\right\} $ with practical estimators have been presented in \citep{wu_2016,jiao_2015}.

Contrasting this set of results, it is well-known that the equivalent problem of estimating the distribution consistently (in total variation) in finite alphabet requieres a sampling complexity that scales like $O(k)$.  Consequently, the large alphabet results for entropy estimation 
show that the task of entropy estimation in finite alphabets is simpler than estimating the high dimensional distribution in terms of sampling complexity.  These findings are consistent with the observation that the entropy is a continuous functional of the space of distributions (in the total variational distance sense) for the finite alphabet case \citep{csiszar_2004,cover_2006,ho_2009,silva_isit_2012}. 

\subsection{Infinite Alphabets}
In this work we are interested in the infinite alphabet scenario, i.e., on the
estimation of the entropy when the alphabet is countable infinite and we have a finite number of samples. 

This problem is an infinite alphabet regime as the size of the alphabet goes unbounded and $n$ 
is kept finite for the analysis,  which contrasts with the large and finite alphabet regime elaborated above.  
As argued in \citep[Sec. IV]{ho_2010},  this is a challenging non-parametric learning problem  because some of the  finite alphabet properties of the entropy do not extend to this infinite dimensional problem.  
Notably, it has been shown that the Shannon entropy is not a continuous functional with respect to the total variational distance in infinite alphabets \citep{harremoes_2007,ho_2009,ho_2010}. 
In  particular,  Ho {\em et al.}  \citep[Theorem 2]{ho_2009} showed concrete examples where convergence in $\chi^2$-divergence and in direct information  divergence  (I-divergence), 
both stronger than total variational convergence \citep{csiszar_2004,devroye_2001}, 
do not imply  convergence of the entropy functional. In addition, \cite{harremoes_2007}   showed  the discontinuity  of the entropy with respect to the reverse  I-divergence \citep{barron_1992},  and consequently,  with respect to the total variational distance\footnote{The  distinction between reverse and direct I-divergence was pointed out in the work of Barron {\em et al.} \citep{barron_1992},  in the context of distribution estimation.}. 
%
In entropy estimation,  the discontinuity of the entropy with respect to the divergence implies that the minimax mean square error is  unbounded\footnote{ $R^*_n = \inf_{\hat{H}(\cdot)} \sup_{\mu \in \mathbb{H}(\mathbb{X}) } \mathbb{E}_{X_1,..X_n\sim \mu^n} \left\{ \left( \hat{H}(X_1,..,X_n) - H(\mu) \right)^2 \right\}=\infty$, where $\mathbb{H}(\mathbb{X})$ denotes the family of finite entropy distribution over the countable alphabet set $\mathbb{X}$. The proof of this result  
follows from \citep[Th. 1]{ho_2010} and  the {Le Cam's} two point method \citep[Sec.2.4.2]{tsybakov_2009}. The argument is presented in Appendix \ref{pro_unbounded_minimax_risk}.}. Consequently, there is no universal minimax  consistent estimator (in the mean square error over sense) of the entropy over the family of finite entropy distributions.

%
Considering  a point-wise or sample wise convergence to zero of the estimation error (instead of the  worse case expected error analysis mentioned above), \cite[Th. 2 and Cor. 1]{antos_2002} showed the remarkable result that the classical plug-in estimate is strongly consistent and consistent in the mean square error sense for any finite entropy distribution, i.e., the simpler  and straightforward  plug-in estimator of entropy is universal, where convergence to the right limiting value is achieved almost surely despite  the discontinuity of the  entropy functional. Moving on the analysis of the (point-wise) rate of convergence of the estimation error,  they presented a finite length lower bound for this error of any arbitrary estimation scheme \citep[Th.3]{antos_2002} showing as a corollary that no universal rate of convergence (to zero) can be achieved for entropy estimation in infinite alphabets \citep[Th.4]{antos_2002}. Finally constrained the problem to a family of distributions with some specific power tail bounded conditions, \cite[Th.7]{antos_2002} show a sharp finite length expression for the rate of convergence of the estimation error of the classical plug-in estimate.

\subsection{From convergence  to Entropy estimation}
Considering the discontinuity of the entropy in $\infty$-alphabets, this work looks at the problem of point-wise almost sure entropy estimation from the new angle of  studying and applying some recent entropy convergence results and their derived bounds in $\infty$-alphabets \citep{piera_2009,silva_isit_2012,ho_2010}.
Entropy convergence results have stablished concrete conditions on both the limiting distribution $\mu$ and the way a sequence of distributions $\left\{\mu_n: n\geq 0\right\}$ convergence to $\mu$ such that  $\lim_{n \rightarrow \infty}  H(\mu_n )=H(\mu)$ is satisfied. The conjecture that motivates this research is that putting these conditions  in the context of a learning task,  i.e., where $\left\{\mu_n: n\geq 0\right\}$ is a random sequence of distributions driven by the classical empirical process,  will offer the possibility to study new family of plug-in estimates with the objective to derive new strong consistency and rate of convergence results under some regularity conditions on $\mu$.  On the practical side, this work explores a data-driven histogram based estimator as a key case of study,  because this approach offers the flexibility to adapt to learning task when appropriate bounds for the estimation and approximation error are derived from the analysis of the problem. 

On the specifics, we begin revisiting the classical plug-in entropy estimator considering the un-explored and relevant scenario where $\mu$ (the data-generated distribution)  has a finite but arbitrary large and unknown support.  This was declared to be a challenging problem in \cite[Th.13]{ho_2010} because of the discontinuity of the entropy.
Finite-length (non-asymptotic) deviation inequalities and intervals of confidence are derived
extending the results presented in \citep[Sec. IV]{ho_2010}. Here it is shown that the classical 
plug-in estimate achieves optimal rates of convergences.  
Relaxing the finite support restriction on $\mu$, we present two histogram-based plug-in estimates, 
one based on 
the celebrated {\em Barron-Gy\"{o}rfi-van der Meulen} estimate \citep{barron_1992, berlinet_1998, vajda_2001};
and the other on a data-driven partition of the space \citep{lugosi_1996_ann_sta,silva_2010b,silva_2010}. 
For the Barron plug-in estimate almost sure consistency is shown for entropy estimation and distribution estimation in direct I-divergence under some mild support conditions on $\mu$.
On the other hand for the data-driven partition scheme,  we show that the estimator is strongly consistent distribution-free,  matching the universal result obtained for the classical plug-in estimate in \citep{antos_2002}. Furthermore, new almost sure rate of convergence results (in the estimation error) are obtained for distributions with finite but unknown support and families of distributions with power and exponential tail dominating conditions. In this context, our result shows that this  adaptive scheme offers optimal and near optimal rate of convergences,  as it approaches arbitrary closely the rate of convergence $O(1/\sqrt{n})$ that is optimal for the finite alphabet problem \citep{vaart_2000}.

The rest of the paper is organized as follows.  Section \ref{sec_pre} introduces some basic concepts,  notation and summarizes the entropy convergence results used in this study. Sections \ref{sec_main}, \ref{subsec:barron} and \ref{subsec:data_driven} state  the main results of this work. The main technical derivation are presented in Section \ref{sec_proofs}, Finally summary and final discussion are given in Section \ref{sec_final}. The proof of some technical auxiliary results are relegated to the Appendix section.

\section{Preliminaries}
\label{sec_pre}
Let $\mathbb{X}$ be a countably  infinite set, without loss of generality the integers, 
and let $\mathcal{P}(\mathbb{X})$ denote the collection of probability measures 
in $\mathbb{X}$. 
For $\mu$ and $v$ in $\mathcal{P}(\mathbb{X})$, 
and $\mu$  absolutely continuous with respect to $v$ (i.e., $\mu \ll v$) \footnote{if  $v(A)=0$ implies that $\mu(A)=0$ for any event $A\in \mathcal{B}(\mathbb{X})$.},  $\frac{d\mu}{d v}(x)$ denotes the  {\em Radon-Nikodym} (RN) derivative  of $\mu$ with respect to $v$. 
Every $\mu \in \mathcal{P}(\mathbb{X})$ is absolutely continuous 
with respect to the counting measure $\lambda$ (or the {\em Lebesgue measure})\footnote{$\lambda(A)=\left|A\right|$ for all $A \subset \mathbb{X}$.}, where its RN derivate is the {\em probability mass function} (pmf), 
$f_{\mu}(x)\equiv \frac{d \mu}{d \lambda}(x)=\mu\left(\left\{x\right\}\right)$, $\forall x \in \mathbb{X}$.  
Finally for any $\mu \in \mathcal{P}(\mathbb{X})$, $A_\mu  \equiv \left\{x\in \mathbb{X}:  f_\mu(x) >0\right\}$ denotes its support and
\begin{equation}\label{eq_pre_1a} 
	 \mathcal{F}(\mathbb{X}) \equiv \left\{\mu \in \mathcal{P}(\mathbb{X}): \lambda(A_\mu) < \infty\right\}
\end{equation}  
denotes the collection of  probabilities with finite support.

Let $\mu$ and $v$ in $\mathcal{P}(\mathbb{X})$, the total variation distance of $\mu$ and $v$ is 
given by \citep{devroye_2001}
\begin{equation}\label{eq_pre_5}
	V(\mu,v) \equiv \sup_{A \in \mathcal{B}(\mathbb{X})} \left|v(A)-\mu(A)\right|,
\end{equation}
where $\mathcal{B}(\mathbb{X})$ is a short-hand for the subsets of $\mathbb{X}$. 
The {\em Kullback-Leibler divergence} or {\em  I-divergence} of $\mu$ with respect to $v$ is given by
\begin{equation}\label{eq_pre_7}
	D(\mu||v) \equiv \sum_{x \in A_\mu} f_\mu(x) \log \frac{f_\mu(x)}{f_v(x)}, 
\end{equation}
when $\mu \ll v$, and $D(\mu||v)$ is  set to infinite, otherwise \citep{kullback_1951}\footnote{The well-known {\em Pinsker`s inequality} offers a relationship 
between the I-divergence and the variational distance 
\citep{kullback_1967,csiszar_1967,kemperman_1969}:  
$\forall \mu,v \in \mathcal{P}(\mathbb{X})$, 
	$2 \ln 2\cdot  V(\mu,v)^2 \leq D(\mu,v)$.
}.

The {\em Shannon entropy} of $\mu\in \mathcal{P}(\mathbb{X})$ is 
given by \citep{cover_2006,gray_1990_b,beirlant_1997}: 
\begin{equation}\label{eq_pre_11}
		{H}(\mu) \equiv  - \sum_{x \in A_\mu} f_\mu(x)  \log f_\mu(x).  
\end{equation}
In this context, it is useful to denote by $\mathbb{H}(\mathbb{X}) \subset \mathcal{P}(\mathbb{X})$
the collection of probabilities where (\ref{eq_pre_11}) is well defined,  by 
$\mathcal{AC}(\mathbb{X}|v) \equiv   \left\{\mu \in \mathcal{P}(\mathbb{X}):  \mu \ll v \right\}$
the collection of measures absolutely continuous with respect to $v \in  \mathcal{P}(\mathbb{X})$, and  by 
$\mathbb{H}(\mathbb{X}|v) \subset \mathcal{AC}(\mathbb{X}|v)$ the collection of probabilities 
where  (\ref{eq_pre_7}) is well defined for $v\in  \mathcal{P}(\mathbb{X})$.

Concerning convergence,  a sequence $\left\{\mu_n: n \in \mathbb{N}\right\} \subset \mathcal{P}(\mathbb{X})$  is said to converge in total variation to $\mu \in \mathcal{P}(\mathbb{X})$ if 
 \begin{equation}\label{eq_sec_pre_11}
 	\lim_{n \rightarrow \infty}  V(\mu_n,\mu)=0.
 \end{equation}
For countable alphabets,  \citep[Lemma 3]{piera_2009}  shows that the convergence in total variation is equivalent to the {\em weak convergence}\footnote{$\left\{\mu_n: n \in \mathbb{N}\right\} \subset \mathcal{P}(\mathbb{X})$  is said to converge  weakly to $\mu \in \mathcal{P}(\mathbb{X})$ if for any bounded function $g(\cdot): \mathbb{X} \rightarrow \mathbb{R}$,  $\lim_{n \rightarrow \infty} \sum_{x\in \mathbb{X}} g(x) f_{\mu_n}(x)= \sum_{x \in \mathbb{X}}g(x) f_{\mu}(x).$},
 which is denoted here by $\mu_n \Rightarrow \mu$, and the point-wise convergence of the  pmf's.  
Furthermore from (\ref{eq_pre_5}),  the convergence in total variation implies the uniform convergence 
of the pmf's, i.e,   $\lim_{n \rightarrow \infty} \sup_{x \in \mathbb{X}} \left| \mu_n(\left\{x\right\})- \mu(\left\{x\right\})\right|=0$.  Therefore in this countable case,  all the four previously mentioned notions of convergence are equivalent:  total variation; weak convergence; point-wise convergence of the pmf`s; 
and uniform convergence of the pmf's. 
 
We conclude with the convergence in I-divergence  introduced by 
\cite{barron_1992}. We say that $\left\{\mu_n: n \in \mathbb{N}\right\}$
converges to $\mu$ in direct and reverse I-divergence if
$\lim_{n \rightarrow \infty}  D(\mu || \mu_n)=0$ and $\lim_{n \rightarrow \infty}  D(\mu_n || \mu)=0$, 
respectively.  From {\em Pinsker's inequality}, the convergence in 
I-divergence implies the  weak convergence in (\ref{eq_sec_pre_11}), where
it is known that the converse 
is not  true \citep{harremoes_2007}.

\subsection{Convergence results for the Shannon entropy}
\label{sec_conv_inq}
The discontinuity of the entropy rises the problem of finding conditions under which convergence of the entropy can be obtained 
in infinite alphabets. On this topic, Ho {\em et al.} \citep{ho_2010} have studied the interplay between the entropy and the  total variation distance stipulating conditions for convergence  by assuming a finite support on the involved distributions.   On the other hand,  \cite[Theorem 21]{harremoes_2007}  obtained convergence of the entropy by imposing a {\em power dominating condition} \cite[Def. 17]{harremoes_2007}  on the limiting probability measure $\mu$, for all the sequences $\left\{\mu_n: n\geq 0\right\}$ converging in reverse I-divergence to $\mu$ \citep{barron_1992}. More recently,   \cite{silva_isit_2012} have addressed the entropy convergence studying a number of new settings that involve  conditions on the limiting measure $\mu$,  as well as the way the sequence $\left\{\mu_n: n\geq 0\right\}$ convergences to $\mu$ in the space of distributions.
This convergence results offer sufficient conditions where the entropy evaluated in a sequence of distribution convergences to the entropy of its limiting distribution and, consequently,  the possibility of applying these results when analyszng plug-in entropy estimator.  The results used in this work are summarized in the rest of this section.

Let us begin with the case when $\mu \in \mathcal{F}(\mathbb{X})$, i.e.,  
the support of the limiting measure is finite and unknown.  
\begin{proposition} \label{lemma_fs_under}
Let us assume that $\mu \in \mathcal{F}(\mathbb{X})$ and $\left\{\mu_n: n \in \mathbb{N}\right\} \subset \mathcal{AC}(\mathbb{X}|\mu)$. 
	If $\mu_n \Rightarrow \mu$,  then $\lim_{n \rightarrow \infty}  D(\mu_n || \mu)=0$
	and $\lim_{n \rightarrow \infty}  H(\mu_n )=H(\mu)$.
\end{proposition}

This result is well-known because when $A_{\mu_n}\subset A_\mu$ for all $n$ the scenario
reduces to the finite alphabet case, in which the entropy is known to be continuous 
\citep{csiszar_2004, cover_2006}. 
Because  we obtain two inequalities used in the rest of the exposition,  a proof is provided here.
\begin{proof}
	$\mu$ and $\mu_n$ belong to $\mathbb{H}(\mathbb{X})$ from the 
	finite-supported assumption. The same reason can be used to show that 
	$D(\mu_n||\mu)<\infty$,  since $\mu_n \ll \mu$ for all $n$. Let us consider the following identity: 
	\begin{equation}\label{eq_lemma_fs_1}
	 	H(\mu )- H(\mu_n) = \sum_{x \in A_\mu} (f_{\mu_n}(x)-f_\mu(x)) \log f_\mu(x) + D(\mu_n||\mu).
	\end{equation}
	The first term on the right hand side (RHS) of (\ref{eq_lemma_fs_1}) is upper bounded 
	by $\mathbf{M}_\mu \cdot V(\mu_n,\mu)$  where
	\begin{equation}\label{eq_lemma_fs_2}
		\mathbf{M}_\mu = \log \frac{1}{\mathbf{m}_\mu} \equiv \sup_{x \in A_\mu}  \left| \log \mu(\left\{x\right\})\right| < \infty.
	\end{equation}
	For the second term,  we have that
	\begin{align}\label{eq_lemma_fs_3} 
		D(\mu_n||\mu) &\leq  \log e \cdot \sum_{x \in A_{\mu_n}}   f_{\mu_n}(x) \left|\frac{f_{\mu_n}(x)}{f_{\mu}(x)}-1\right| \nonumber\\ 
				      &\leq  \frac{\log e}{\mathbf{m}_\mu} \cdot \sup_{x\in A_\mu} \left|f_{\mu_n}(x) - f_{\mu}(x)\right| 
				      \leq  \frac{\log e}{\mathbf{m}_\mu} \cdot V(\mu_n,\mu).	
	\end{align}
	and, consequently,  
	\begin{align}\label{eq_lemma_fs_4} 
		\left|H(\mu )- H(\mu_n)\right| \leq 
		\left[ \mathbf{M}_\mu+\frac{\log e}{\mathbf{m}_\mu}\right] \cdot V(\mu_n,\mu).
	\end{align} 	
\end{proof}
Under the assumptions of Proposition \ref{lemma_fs_under}, we note that
the reverse I-divergence and  the entropy difference
are bounded by the total variation by  (\ref{eq_lemma_fs_3}) and (\ref{eq_lemma_fs_4}), respectively. Note however that these bounds are distribution dependent function of $\mathbf{m}_\mu$($\mathbf{M}_\mu$) in (\ref{eq_lemma_fs_2}).\footnote{It is simple to note that $\mathbf{m}_\mu(\mathbf{M}_\mu)<\infty$ if,  and only if, $\mu\in \mathcal{F}(\mathbb{X})$.} The next result relaxes the assumption that  $\mu_n \ll \mu$ 
and offers 
a necessary and sufficient condition for the convergence of the entropy.
\begin{lemma}\label{fs_main_th} \citep[Th. 1]{silva_isit_2012}
	Let $\mu \in \mathcal{F}(\mathbb{X})$ and $\left\{\mu_n: n \in \mathbb{N}\right\} \subset \mathcal{F}(\mathbb{X})$.   	If 
	$\mu_n \Rightarrow  \mu$,  then
		there exists $N>0$ such that $\mu \ll \mu_n$ 
		$\forall n\geq N$,
		 and
		\begin{equation} 
			\lim_{n \rightarrow \infty}  D(\mu || \mu_n)=0. \nonumber
		\end{equation}
		Furthermore, $\lim_{n \rightarrow \infty}  H(\mu_n )=H(\mu)$, if and only if \footnote{$\mu(\cdot |B)$ in (\ref{eq_lemma_fs_9}) denotes the conditional probability 
	of $\mu$ given the event $B \subset \mathbb{X}$}, 
		\begin{equation}\label{eq_lemma_fs_9} 
			\lim_{n \rightarrow \infty}  \mu_n \left(A_{\mu_n} \setminus A_{\mu}\right)  \cdot H\left(\mu_n\left(\cdot |A_{\mu_n} \setminus A_{\mu}\right)\right)=0  \Leftrightarrow
			\lim_{n \rightarrow \infty} \sum_{x\in A_{\mu_n} \setminus A_{\mu}} f_{\mu_n}(x) \log \frac{1}{f_{\mu_n}(x)}=0.
		\end{equation}
\end{lemma}
	Lemma \ref{fs_main_th} tells us that  in order to achieve entropy convergence (on top of the weak convergence), it is necessary and sufficient to ask for a vanishing  expression (with $n$) of the entropy of $\mu_n$ restricted to the elements of the set $A_{\mu_n} \setminus A_{\mu}$.  
	Two remarks about this result: 
	{\bf 1)} The convergence in direct $I$-divergence 
	does not imply the convergence of the entropy\footnote{Concrete examples are presented 
	in \citep[Sec III]{ho_2009} and \citep{silva_isit_2012}.}.
	{\bf 2)} Under the assumption that $\mu \in \mathcal{F}(\mathbb{X})$,  $\mu$ is eventually 
	absolutely  continuous with respect to $\mu_n$, and the convergence in total variations is equivalent to the convergence in direct I-divergence.

We conclude this section with the case when the support of $\mu$ is
infinite and unknown, i.e., $\left|A_\mu\right|=\infty$.  
In this context, we highlight two results:
\begin{lemma}\citep[Theorem 4]{piera_2009} \label{th_con_infty_1}
	Let us consider that  $\mu \in \mathbb{H}(\mathbb{X})$ and  
	$\left\{\mu_n: n\geq 0\right\} \subset \mathcal{AC}(\mathbb{X}|\mu)$.
	If $\mu_n \Rightarrow  \mu$ and 
	\begin{equation}\label{eq_sec_infty_sp_1}
		M \equiv  \sup_{n\geq 1} \sup_{x\in A_{\mu}} \frac{f_{\mu_n}(x)}{f_\mu(x)} < \infty,
	\end{equation}
	then, $\mu_n \in  \mathbb{H}(\mathbb{X}) \cap \mathbb{H}(\mathbb{X}|\mu)$ for all $n$ and
	it follows that
	\begin{equation}
		\lim_{n \rightarrow \infty}  D(\mu_n || \mu)=0 \ \text{and} \ \lim_{n \rightarrow \infty}  H(\mu_n )=H(\mu).  \nonumber
	\end{equation}
\end{lemma}
Interpreting Lemma \ref{th_con_infty_1}, we have that to obtain the convergence of the entropy functional without imposing a finite support assumption on $\mu$,  a uniform bounding condition (UBC) $\mu$-almost everywhere was added in (\ref{eq_sec_infty_sp_1}).  This UBC allows the use of  the {\em dominated convergence theorem} \citep{varadhan_2001,breiman_1968},  and it is strictly needed in that sense \citep{piera_2009}.
Finally by adding this UBC,  the convergence on reverse I-divergence is also obtained as a byproduct.  

Finally,  when $\mu \ll \mu_n$ for all $n$, we consider the following result:
\begin{lemma} \citep[Theorem 3]{silva_isit_2012} \label{th_con_infty_2}
	Let $\mu \in \mathbb{H}(\mathbb{X})$ and a  sequence of measures 
	$\left\{\mu_n: n\geq 1 \right\} \subset \mathbb{H}(\mathbb{X})$ such that 
	$\mu \ll \mu_n$ for all $n \geq 1$. If $\mu_n \Rightarrow  \mu$
	and 
	\begin{equation}\label{eq_sec_infty_sp_2}
		 \sup_{n\geq 1} \sup_{x\in A_{\mu}} \left|\log \frac{f_{\mu_n}(x)}{f_\mu(x)}\right|  < \infty
	\end{equation}
	then, $\mu \in \mathbb{H}(\mathbb{X}|\mu_n)$ for all $n\geq 1$, and
	\begin{equation}
		\lim_{n \rightarrow \infty}  D(\mu || \mu_n)=0. \nonumber
	\end{equation}
	Furthermore, $\lim_{n \rightarrow \infty}  H(\mu_n )=H(\mu)$, if and only if,  
	\begin{equation}\label{eq_sec_infty_sp_3}
		\lim_{n \rightarrow \infty} \sum_{x \in A_{\mu_n}\setminus A_{\mu}} f_{\mu_n}(x) \log \frac{1}{f_{\mu_n}(x)}=0. 
	\end{equation}
\end{lemma}
	This result shows again the non-sufficiency  
	of the convergence in direct I-divergence to achieve entropy convergence in the regime 
	when $\mu \ll \mu_n$. In fact, 
	Lemma \ref{th_con_infty_2} may be interpreted as an extension of 
	Lemma \ref{fs_main_th} when we relax the finite support assumption on $\mu$.

\section{Shannon entropy estimation}
\label{sec_main}
Let $\mu$ be a probability  in $\mathbb{H}(\mathbb{X})$, and 
let us denote by $X_1,X_2,X_3,\ldots$ the empirical process induced 
from i.i.d. realizations  of a random variable driven by $\mu$, 
i.e., $X_i \sim \mu$, for all $i\geq 0$.  
Let $\mathbb{P}_{\mu}$ denote the distribution of the empirical
process in $(\mathbb{X}^{\infty},\mathcal{B}(\mathbb{X}^{\infty}))$
and $\mathbb{P}^n_{\mu}$ denote the finite block distribution of $X^n_1\equiv (X_1,..., X_n)$
in the product space $(\mathbb{X}^{n},\mathcal{B}(\mathbb{X}^{n}))$.
Given a realization of $X_1,X_2,X_3, \ldots, X_n$, we can construct an 
histogram-based estimator like classical empirical probability given by:
\begin{equation}\label{eq_sec_stat_learn_1}
	\hat{\mu}_n(A) \equiv \frac{1}{n} \sum_{k=1}^n \mathbb{1}_{A}(X_k), \ \forall A \subset \mathbb{X},
\end{equation}
with pmf denoted by $f_{\hat{\mu}_n}(x) 
= \hat{\mu}_n(\left\{x\right\})$ for all $x \in \mathbb{X}$.
A natural estimator of the entropy  is the {\em plug-in estimate} 
of $\hat{\mu}_n$ given by 
\begin{equation}\label{eq_sec_stat_learn_2}
	H(\hat{\mu}_n) = - \sum_{x \in \mathbb{X}} f_{\hat{\mu}_n}(x) \log f_{\hat{\mu}_n}(x), 
\end{equation}
which is a measurable function of $X_1, \ldots, X_n$.\footnote{This dependency on the data will be 
implicit for the rest of the exposition.}

For the rest of the exposition, we use the convergence results in Section \ref{sec_conv_inq} to derive 
strong consistency results for plug-in histogram-based estimates, 
like $H(\hat{\mu}_n)$ in (\ref{eq_sec_stat_learn_2}), as well as finite length concentration inequalities 
to obtain almost-sure rate of convergence for the estimation error  $\left|H(\hat{\mu}_n)  - H({\mu}) \right|$. 
 
\subsection{Revisiting the classical  Plug-in estimator for finite and unknown supported distributions}
\label{sub_sec_fs_class_hist}
We start analyzing the case where $\mu$ has a finite and unknown support. 
A  consequence of {\em the strong law of large numbers} \citep{breiman_1968, varadhan_2001} 
is  that $\forall x \in \mathbb{X}$, $\lim_{n  \rightarrow \infty} 
\hat{\mu}_n(\left\{x\right\})= {\mu}(\left\{x\right\})$, $\mathbb{P}_\mu$-almost surely (a.s.), 
hence $\lim_{n  \rightarrow \infty} V(\hat{\mu}_n, \mu)=0$,  $\mathbb{P}_\mu$-a.s. 
On the other hand,  it is clear that $A_{\hat{\mu}_n} \subset A_{\mu}$ with probability one.
Then adopting  Proposition \ref{lemma_fs_under} it follows that
\begin{align}\label{eq_sec_stat_learn_3} 
&\lim_{n \rightarrow \infty}  D(\hat{\mu}_n  || \mu)=0 \text{ and }
\lim_{n \rightarrow \infty}  H(\hat{\mu}_n )=H(\mu),  \ \mathbb{P}_\mu-a.s., 
\end{align}
i.e., $\hat{\mu}_n$ is a strongly consistent estimator of $\mu$ in reverse I-divergence 
and $H(\hat{\mu}_n)$ is a strongly consistent estimate of $H(\mu)$ distribution-free in $\mathcal{F}(\mathbb{X})$. 
Furthermore, we can state the following: 

\begin{theorem}\label{th_con_inq}
Let $\mu \in \mathcal{F}(\mathbb{X})$  and $\hat{\mu}_n$ be  in (\ref{eq_sec_stat_learn_1}).
Then $\hat{\mu}_n \in \mathcal{H}(\mathbb{X}) \cap \mathcal{H}(\mathbb{X}|\mu)$, 
$\mathbb{P}_\mu$-a.s and $\forall n \geq 1$,
$\forall \epsilon>0$,
	\begin{align} \label{eq_sec_stat_learn_6} 
		&\mathbb{P}^n_{\mu} \left(D(\hat{\mu}_n || \mu) > \epsilon\right) \leq 2^{\left|A_{\mu}\right|+1} \cdot e^{- \frac{2 \mathbf{m}_{\mu}^2 \cdot  n \epsilon^2}{\log  e ^2 }}, \\
	\label{eq_sec_stat_learn_6b} 
		& \mathbb{P}^n_{\mu} \left( \left|H(\hat{\mu}_n)  - H({\mu}) \right|  > \epsilon\right) \leq 2^{\left|A_{\mu}\right| +1} \cdot e^{-\frac{2 n \epsilon ^2}{(\mathbf{M}_{\mu} + \frac{\log e}{\mathbf{m}_{\mu}})^2}}.
	\end{align}
	Moreover, $D(\mu||\hat{\mu}_n)$  is eventually well-defined with probability one, 
	and $\forall \epsilon>0$, and for any $n \geq 1$, 
	\begin{align} \label{eq_sec_stat_learn_12b} 
		&\mathbb{P}^n_{\mu} \left(D(\mu || \hat{\mu}_n) > \epsilon\right) \leq
		2^{\left|A_\mu\right|+1} \cdot
		\left[   e^{-\frac{2 n \epsilon^2}{ \log e^2  \cdot \left({1}/{\mathbf{m}_u}+1\right)^2}}
		+ e^{- n \mathbf{m}_\mu^2 } \right].
	\end{align}
\end{theorem}
%
This result  implies that  for any $\tau \in (0,1/2)$ and $\mu \in \mathcal{F}(\mathbb{X})$, $\left|H(\hat{\mu}_n)  - H({\mu}) \right|$,  $D(\hat{\mu}_n ||\mu )$ and  $D(\mu || \hat{\mu}_n)$ goes to zero as $o(n^{-\tau})$ $\mathbb{P}_\mu$-a.s.  Furthermore,   $\mathbb{E}_{\mathbb{P}^n_{\mu}} \left( \left|H(\hat{\mu}_n)  - H({\mu}) \right| \right)$ and  $\mathbb{E}_{\mathbb{P}^n_{\mu}}(D(\hat{\mu}_n || \mu))$  behave like $O(1/\sqrt{n})$ for all $\mu \in \mathcal{F}(\mathbb{X})$ from  (\ref{eq_sec_stat_learn_4}) in Sec. \ref{sub_sec_th_con_inq}, which is the optimal  rate of convergence of the finite alphabet scenario.  As a direct corollary of (\ref{eq_sec_stat_learn_6b}), 
it is possible to derive intervals of confidence for the estimation error 
$\left|H((\hat{\mu}_n)- H(\mu_n)\right|$:  for all $\delta>0$ and $n\geq 1$,  
\begin{align}
	\label{eq_sec_stat_learn_6d} 
	\mathbb{P}_\mu  \left( \left|H((\hat{\mu}_n)- H(\mu_n)\right| \leq \left(\mathbf{M}_{\mu} + {\log e}/{\mathbf{m}_{\mu}}\right)  \sqrt{\frac{1}{2n} \ln   \frac{2^{\left| A_\mu \right| + 1}}{\delta} } \right)  \geq 1-\delta.
\end{align}
This confidence interval 
behaves like $O(1/\sqrt{n})$ as a function of $n$, and like $O(\sqrt{\ln 1/\delta})$ as a function of  $\delta$, which are the same optimal asymptotic trend that can be  obtained for $V(\mu, \hat{\mu}_n)$ in (\ref{eq_sec_stat_learn_4}).

To conclude,  we note that $A_{\hat{\mu}_n}\subset A_{\mu}$ $\mathbb{P}^n_{\mu}$-a.s. where for any  $n\geq 1$, 
	$\mathbb{P}^n_{\mu} (A_{\hat{\mu}_n}\neq A_{\mu})>0$  
implying that $\mathbb{E}_{\mathbb{P}^n_{\mu} }(D(\mu||\hat{\mu}_n))=\infty$, $\forall n$. 
Then even in the finite and unknown supported scenario, $\hat{\mu}_n$ is not  consistent 
in expected direct I-divergence, which is congruent with the result in  \citep{barron_1992,gyorfi_1994b}.
Besides this negative result,  strong consistency in direct $I$-divergence can be obtained from 
(\ref{eq_sec_stat_learn_12b}), in the sense that $\lim_{n \rightarrow \infty} D(\mu || \hat{\mu}_n)=0$, $\mathbb{P}_{\mu}$-a.s.

\subsection{A simplified version of the Barron estimator for finite supported measures}
\label{sub_sec_barron_simplified}
It is well-understood  that consistency in  expected direct  I-divergence is  of critical importance 
for the construction of a lossless universal source coding scheme \citep{barron_1992,gyorfi_1994b,csiszar_2004,cover_2006,rissanen_2010}. 
Here we explore an estimator that achieves this learning objective in addition to entropy estimation.

For that, let $\mu \in \mathcal{F}(\mathbb{X})$ and let us assume that 
we know a measure $v\in \mathcal{F}(\mathbb{X})$ such that 
$\mu \ll v$. \cite{barron_1992} proposed a modified version 
of the empirical measure in (\ref{eq_sec_stat_learn_1}) to estimate $\mu$
from i.i.d. realizations,  adopting a mixture estimate of the form 
\begin{align} \label{eq_sec_stat_learn_12} 
	\tilde{\mu}_n(B) = (1-a_n)\cdot \hat{\mu}_n(B) + a_n \cdot v(B),
\end{align}
for all $B\subset \mathbb{X}$, and with $(a_n)_{n\in \mathbb{N}}$ a sequence of real
numbers in $(0,1)$.
Note that $supp(\tilde{\mu}_n)=A_{v}$ then $\mu \ll \tilde{\mu}_n$ for all $n$\footnote{From 
the finite support assumption $H(\tilde{\mu}_n) <\infty$ and $D(\mu || \tilde{\mu}_n) < \infty$, $\mathbb{P}_{\mu}$-a.s.}. 
The following result derives from the convergence result in Lemma \ref{fs_main_th}.

\begin{theorem}\label{pro_barron_conv}
	Let  $v \in \mathcal{F}(\mathbb{X})$ and $ \mu \in \mathcal{AC}(\mathbb{X}| v)$, and let  us 
	consider $\tilde{\mu}_n$  in (\ref{eq_sec_stat_learn_12}), with respect to $v$, 
	induced from i.i.d. realizations of $\mu$.
	\begin{itemize}
	\item [i)] If  $(a_n)$ is $o(1)$, then  
		$\lim_{n \rightarrow \infty}  H(\tilde{\mu}_n )=H(\mu)$,  
 		$\lim_{n \rightarrow \infty}  D(\mu || \tilde{\mu}_n)=0$, $\mathbb{P}_{\mu}-a.s.$, and 
		$\lim_{n \rightarrow \infty}  \mathbb{E}_{\mathbb{P}_{\mu}}(D(\mu || \tilde{\mu}_n))=0$. 
	\item [ii)] Furthermore, if $(a_n)$ is $O(n^{-p})$ with $p > 2$, then for all $\tau\in (0,1/2)$,  $\left| H(\tilde{\mu}_n)-H(\mu) \right|$  and  $D(\mu||\tilde{\mu}_n)$ are $o(n^{-\tau})$ $\mathbb{P}_{\mu}$-a.s,
	and  $\mathbb{E}_{\mathbb{P}_{\mu}}( \left| H(\tilde{\mu}_n)-H(\mu) \right|)$ and 
	$\mathbb{E}_{\mathbb{P}_{\mu}}(D(\mu||\tilde{\mu}_n))$ are $O(1/\sqrt{n})$.
	\end{itemize}
\end{theorem}

\section{The Barron-Gy\"{o}rfi-van der Meulen Estimator}
\label{subsec:barron}
The celebrated Barron estimate was  proposed by Barron, Gy\"{o}rfi and {van der Meulen} \citep{barron_1992} in the context of an abstract and in general continuous measurable space. It was designed as a variation of the classical histogram-based scheme to achieve a consistent estimate of the distribution in 
direct I-divergence \citep[Theorem 2]{barron_1992}.\footnote{As mentioned before,  consistency in expected direct I-divergence is an important learning topic because  of its connection with lossless universal  source coding \citep{gyorfi_1994b,cover_2006, barron_1992}, 
where it is well-known that there is no distribution-free consistent estimate in direct I-divergence 
in the  infinite alphabet case \citep{gyorfi_1994b}.} 
Here we revisit the {\em Barron estimate} in our countable alphabet scenario,  with the objective of estimating the Shannon entropy consistently, which to the best of our knowledge has not been previously addressed in the literature. For that purpose, Lemma \ref{th_con_infty_2} will be used as a key result.

Let $v \in \mathcal{P}(\mathbb{X})$ of infinite support (i.e., $\mathbf{m}_v=\inf_{x\in A_v} v(\left\{x\right\})=0$).
We want  to construct a strongly consistent estimate of the entropy 
restricted to the collection of probabilities in  $\mathbb{H}(\mathbb{X}|v)$. For that, let us consider a sequence $(h_n)_{n\geq 0}$ with values in $(0,1)$ and let us denote by $\pi_n=\left\{A_{n,1},A_{n,2},\ldots, A_{n,m_n}\right\}$ the finite partition of $\mathbb{X}$  with maximal cardinality satisfying that
\begin{equation}\label{eq_subsec:barron_1}
	v(A_{n,i}) \geq h_n, \ \forall i \in \left\{1,..,m_n\right\}.
\end{equation}
Note that $m_n=\left|\pi_n\right|\leq 1/h_n$ for all $n\geq 1$, 
and because of  the fact that $\inf_{x\in A_v} v(\left\{x\right\})=0$
it is simple to verify that  if $(h_n)$ is $o(1)$ and then $\lim_{n \rightarrow \infty} m_n=\infty$.
Note that $\pi_n$ offers an approximated statistically equivalent partition of $\mathbb{X}$ with respect to the reference measure $v$. 
In this context,  given $X_1,\ldots,X_n$,  i.i.d. realizations of 
$ \mu \in \mathbb{H}(\mathbb{X}|v)$,  
the idea proposed by \cite{barron_1992} was to estimate the RN derivative $\frac{d \mu}{d v}(x)$ 
by the following  histogram-based construction: 
\begin{equation}\label{eq_subsec:barron_2}
	\frac{d \mu^*_n}{d v} (x) = (1-a_n) \cdot  \frac{\hat{\mu}_n(A_n(x))}{v(A_n(x))} + a_n, \ \forall x \in A_v,
\end{equation}
where $a_n$ is a real number in $(0,1)$, $A_n(x)$ denotes the cell in $\pi_n$
that contains the point $x$, and $\hat{\mu}_n$ is the  empirical measure in (\ref{eq_sec_stat_learn_1}).
Note that $$f_{\mu^*_n}(x)=\frac{d \mu^*_n}{d \lambda}(x)=f_v(x)\cdot \left[(1-a_n) \cdot  \frac{\hat{\mu}_n(A_n(x))}{v(A_n(x))} + a_n\right],$$ $\forall x \in \mathbb{X}$, 
and, consequently,  $\forall B \subset \mathbb{X}$
\begin{equation}\label{eq_subsec:barron_2b}
	\mu^*_n(B)=(1-a_n) \sum_{i=1}^{m_n} \hat{\mu}_n(A_{n,i})\cdot \frac{v(B\cap A_{n,i})}{v(A_{n,i})} 
			+  a_n  v(B).
\end{equation}
By construction $A_\mu\subset A_v \subset supp(\mu^*_n) $ and, consequently, 
$\mu \ll \mu^*_n$ for all $n\geq 1$.
The next result shows sufficient conditions on the sequences $(a_n)$ and $(h_n)$ 
to guarantee a strongly consistent estimate  of the entropy $H(\mu)$ and of $\mu$ in direct $I$-divergence,  distribution  free in $\mathbb{H}(\mathbb{X}|v)$.  The proof is based on verifying that the sufficient conditions of Lemma \ref{th_con_infty_2} are satisfied $\mathbb{P}_{\mu}$-a.s. 
\begin{theorem}\label{th_barron}
	Let $v$ be in  $\mathcal{P}(\mathbb{X}) \cap  \mathbb{H}(\mathbb{X})$ 
	with infinite support,  and let us consider $\mu$ in $\mathbb{H}(\mathbb{X}|v)$. 
	If we have that:
	\begin{itemize}
		\item[i)]  $(a_n)$ is $o(1)$ and $(h_n)$ is $o(1)$,
		\item[ii)] $\exists \tau \in (0,1/2)$, such that the sequence  $\left(\frac{1}{a_n\cdot h_n}\right)$ is $o(n^\tau)$,
	\end{itemize}
	then $\mu \in \mathbb{H}(\mathbb{X}) \cap \mathbb{H}(\mathbb{X}|\mu^*_n)$ 
	for all $n\geq 1$ and
	\begin{equation}\label{eq_subsec:barron_3}
		\lim_{n \rightarrow \infty}  H( \mu^*_n)=H(\mu)  \text{ and } \lim_{n \rightarrow \infty}  D(\mu || \mu^*_n)=0, \ \mathbb{P}_\mu-a.s..
	\end{equation}
\end{theorem}

The Barron estimator \citep{barron_1992}  was originally proposed 
in the context of distributions defined in an abstract measurable space.  Then if we restrict \citep[Theorem 2]{barron_1992} to our countable alphabet case, the following result is  obtained:
\begin{corollary} \citep[Theorem 2]{barron_1992}
\label{cor_barron}
Let  us consider $v \in \mathcal{P}(\mathbb{X})$ and  $\mu \in \mathbb{H}(  \mathbb{X}|v)$.  If $(a_n)$ is $o(1)$, $(h_n)$ is $o(1)$ and  $\lim \sup_{n \rightarrow \infty} \frac{1}{n a_n h_n} \leq 1$ 
then 
$$\lim_{n \rightarrow \infty} D(\mu ||\mu^*_n)=0,\  \mathbb{P}_\mu-a.s.$$
\end{corollary}  
Therefore, when the objective is the estimation of distributions consistently in direct I-divergence,  Corollary \ref{cor_barron} should be considered to be a better result\footnote{Corollary \ref{cor_barron} offers weaker conditions than Theorem \ref{th_barron} (in particular condition ii)).}. On the other hand, the proof of Theorem \ref{th_barron} is based on verifying the sufficient conditions of Lemma \ref{th_con_infty_2},  where the objective is to achieve the convergence of the entropy,  and as a consequence, the  convergence in direct I-divergences. Therefore,  we can say that the stronger conditions of Theorem \ref{th_barron} are needed when the objective is entropy estimation. This can be justified from the fact that convergence in direct I-divergence does not imply entropy convergence in the countable case,  as was discussed in Section \ref{sec_conv_inq} (see, Lemmas \ref{fs_main_th} and \ref{th_con_infty_2}).

\section{A Data-Driven histogram-based estimator}
\label{subsec:data_driven}
Data-driven partitions can approximate better the nature of the empirical distribution in the sample space with few quantization bins \citep{nobel_1996b}. They has the flexibility to improve the approximation quality of histogram-based estimates and from that obtain better  performances in different non-parametric learning tasks \citep{lugosi_1996_ann_sta, silva_2010,silva_2010b,silva_2012,darbellay_1999}.  
One of the basic principle of this approach  is to partition $\mathbb{X}$ into data-dependent cells in order to preserve a critical number of samples per cell. This last condition will be crucial  to derive a compromise between an estimation and approximation error that will be used in the proof of two of the main results of this section (Theorems \ref{th_rate_data_driven_power} and \ref{th_rate_data_driven_exp_family}).
 
Given $X_1,..,X_n$ i.i.d. realizations driven by $\mu \in \mathbb{H}(\mathbb{X})$ and $\epsilon>0$, 
let us define the data-driven set
\begin{equation}\label{eq_data_driven_1}
	\Gamma_\epsilon \equiv \left\{x \in \mathbb{X}: \hat{\mu}_n( \left\{x \right\})\geq \epsilon \right\}, 
\end{equation}
and $\phi_\epsilonÊ\equiv \Gamma_\epsilon^c$. Let $\Pi_\epsilon \equiv \left\{  \left\{x \right\}: x \in \Gamma_\epsilon \right\} \cup  \left\{ \phi_\epsilon\right\} \subset \mathcal{B}(\mathbb{X})$ be a data-driven partition  with maximal resolution in $\Gamma_\epsilon$, and $\sigma_\epsilon \equiv \sigma(\Pi_\epsilon)$ be the smallest sigma field that contains  $\Pi_\epsilon$
 \footnote{As $\Pi_\epsilon$ is a finite partition,  $\sigma_\epsilon$ is the collection of sets that are union of elements  of $\Pi_\epsilon$.}. We propose  the conditional 
empirical measure restricted to  $\Gamma_\epsilon$ by:
\begin{equation}\label{eq_data_driven_2}
	\hat{\mu}_{n,\epsilon} \equiv \hat{\mu}_n(\cdot |\Gamma_\epsilon). 
\end{equation}
Note that by construction $supp(\hat{\mu}_{n,\epsilon})=\Gamma_\epsilon \subset A_{\mu}$, $\mathbb{P}_\mu$-a.s. and  consequently $\hat{\mu}_{n,\epsilon} \ll \mu$ for all $n\geq 1$. Furthermore, 
$\left| \Gamma_\epsilon \right| \leq \frac{1}{\epsilon}$  and importantly in the context of the entropy functional we have that
\begin{equation}\label{eq_data_driven_3}
	\mathbf{m}_{\hat{\mu}_n}^\epsilon \equiv \inf_{x\in \Gamma_\epsilon} \hat{\mu}_n(\left\{x \right\})\geq \epsilon. 
\end{equation}

The next result establishes a mild sufficient condition on $(\epsilon_n)$ 
 for which the plug-in estimate $H(\hat{\mu}_{n,\epsilon_n})$ is strongly consistent distribution-free in  $\mathbb{H}(\mathbb{X})$.  Considering that  we are in the regime where $\hat{\mu}_{n,\epsilon_n} \ll \mu$, $\mathbb{P}_\mu$-a.s.,  the proof  uses Lemma \ref{th_con_infty_1} as a  central result.

\begin{theorem}\label{th_data_driven}
	 If $(\epsilon_n)$ is $O(n^{-\tau})$ with $\tau\in (0,1)$, then for all $\mu \in \mathbb{H}(\mathbb{X})$
	\begin{equation}
	\lim_{n \rightarrow \infty}  H(\hat{\mu}_{n,\epsilon_n})=H(\mu), 
	\  \mathbb{P}_\mu-a.s. \nonumber
	\end{equation}
\end{theorem}
This theorem tells us that this construction offers a universal estimation of the entropy in the strong 
almost sure sense,  as long as the design parameter $\tau$ belongs to $(0,1)$. 

Complementing Theorem \ref{th_data_driven}, 
the next result offers almost sure rates of converge for a  family of distributions with a power tail bounded condition (TBC).  
In particular, we consider the family of distributions studied by \cite[Th.7]{antos_2002} in the context of characterizing the rate of convergences for the classical plug-in estimate.
\begin{theorem}\label{th_rate_data_driven_power}
	 Let us assume that for some $p>1$ there are two constants
	$0<k_0 \leq k_1$ such that $k_0\cdot x^{-p}\leq \mu ( \left\{ x\right\})\leq k_1 x^{-p}$ for all $x\in \mathbb{X}$. If we consider that $(\epsilon_n)=(n^{-\tau^*})$ for $\tau^*=\frac{1}{2+1/p}$, then 
	\begin{equation} 
	\left|  H(\mu) - H(\hat{\mu}_{n,\epsilon_n})  \right|  \text{ is } O(n^{- \frac{1-{1}/{p}}{2+1/p}}), \ \mathbb{P}_\mu-a.s.\nonumber
	\end{equation}
\end{theorem}
This result states that under the  mentioned $p$-power TBC on $f_{\mu}(\cdot)$, the plug-in estimate $H(\hat{\mu}_{n,\epsilon_n})$  can offer a rate  of convergence to the true limit that is $O(n^{- \frac{1-{1}/{p}}{2+1/p}})$ with probability one.  
The proof of this results set the approximation sequence $(\epsilon_n)$ function of $p$,  by finding an optimal tradeoff between estimation and approximation errors while performing a finite length (non-asymptotic) analysis of the expression $\left|  H(\mu) - H(\hat{\mu}_{n,\epsilon_n})  \right|$ (the details of this analysis are presented in Section \ref{proof_th_rate_data_driven_power}). It is insightful to look at two extreme regimes in this result:  $p$ approaching $1$, in which the rate is arbitrarily slow (approaching a non-decaying behavior); and  $p \rightarrow \infty$, where $\left|  H(\mu) - H(\hat{\mu}_{n,\epsilon_n})  \right|$ is $O(n^{-q})$  for all $q\in (0,1/2)$ $\mathbb{P}_\mu$-a.s..
 This last power decaying range $q\in (0,1/2)$ matches what can be achieved for the finite alphabet scenario in Theorem \ref{th_con_inq} (see  Eq.(\ref{eq_sec_stat_learn_6b})) which it is known to be optimal rate  for finite alphabets. 

Extending Theorem \ref{th_rate_data_driven_power}, the following result addresses  the more
constrained case of distributions with an exponential TBC.
\begin{theorem}\label{th_rate_data_driven_exp_family}
Let us consider $\alpha>0$ and let us assume that there are $k_0,k_1$ with $0<k_0 \leq k_1$ and $N>0$ such that  $k_0 \cdot e^{-\alpha x} \leq \mu ( \left\{ x\right\})\leq k_1 \cdot e^{-\alpha x}$ for all $x \geq N$.  If we consider $(\epsilon_n)=(n^{-\tau})$ with  $\tau \in (0,1/2)$, then
\begin{equation}
	\left|  H(\mu) - H(\hat{\mu}_{n,\epsilon_n})  \right|  \text{ is } O(n^{-\tau} \log n), \ \mathbb{P}_\mu-a.s.\nonumber
\end{equation}
\end{theorem}
Under this stringer TBC on $f_\mu(\cdot)$,  we note that 
	$\left|  H(\mu) - H(\hat{\mu}_{n,\epsilon_n})  \right|  \text{ is } o(n^{-q})$  $\mathbb{P}_\mu-a.s.$,  
for any arbitrary $q \in (0,1/2)$,  by selecting $(\epsilon_n)=(n^{-\tau})$ with  $q<\tau<1/2$. This last condition on $\tau$ is universal over $\alpha>0$.
Remarkably for any distribution with this exponential TBC, we can approximate (arbitrarely closely)  the optimal almost sure rate of convergence achieved for the finite alphabet scenario.  

Finally,  we revisit the finite and unknown supported scenario,  where we show that the data-driven estimate exhibit the same optimal almost sure convergence rate performance of the classical plug-in entropy estimate studied in  Sec. \ref{sub_sec_fs_class_hist}.
\begin{theorem}\label{th_rate_data_driven_finite_support}
	Let us assume that $\mu\in \mathcal{F}(\mathcal{X})$ and $(\epsilon_n)$ being $o(1)$.  
	Then for all $\epsilon>0$ there is $N>0$ such that  $\forall n\geq N$
	\begin{equation} \label{eq_data_driven_6}
	\mathbb{P}^n_{\mu} \left( \left|H(\hat{\mu}_{n,\epsilon_n})  - H({\mu}) \right|  > \epsilon\right) \leq 2^{\left|A_\mu\right|+1} \cdot
		\left[   e^{-\frac{2 n \epsilon ^2}{(\mathbf{M}_{\mu} + \frac{\log e}{\mathbf{m}_{\mu}})^2}}
		+ e^{- \frac{n \mathbf{m}_\mu^2}{4}} \right].
	\end{equation}
\end{theorem}
The proof of this result reduces to verify that $\hat{\mu}_{n,\epsilon_n}$ detects $A_\mu$ almost surely when $n$ goes to infinity and from this, it follows that $H(\hat{\mu}_{n,\epsilon_n})$ matches the optimal almost sure performance of $H(\hat{\mu}_n)$ under the assumption that $\mu\in \mathcal{F}(\mathcal{X})$.  In particular,  (\ref{eq_data_driven_6}) implies that $\left|H(\hat{\mu}_{n,\epsilon_n})  - H({\mu}) \right|$ is $o(n^{-q})$ almost surely for all $q\in (0,1/2)$ as long as $\epsilon_n \rightarrow 0$  with $n$. 

\section{Proofs of the Main results}
\label{sec_proofs}
\subsection{Theorem \ref{th_con_inq}:}
\label{sub_sec_th_con_inq}
\begin{proof}
	Let  $\mu$ be in $\mathcal{F}(\mathbb{X})$,  then $\left|A_{\mu}\right|  \leq k$ for some $k>1$.  
	From  {\em Hoeffding's inequality}  \citep{devroye_2001} 
	 $\forall n\geq 1$, and for any $\epsilon>0$, 
	\begin{align} \label{eq_sec_stat_learn_4} 
		\mathbb{P}^n_{\mu} \left(V(\hat{\mu}_n, \mu) > \epsilon\right) \leq 2^{k+1} \cdot e^{- 2n \epsilon^2}
	\text{ and }
		\mathbb{E}_{\mathbb{P}^n_{\mu}}(V(\hat{\mu}_n, \mu)) \leq 2 \sqrt{\frac{{(k+1)} {\log 2} }{{n}}}.
	\end{align}
	Considering that $\hat{\mu}_n\ll \mu$ $\mathbb{P}_\mu$-a.s, we can use 
	Proposition \ref{lemma_fs_under} to obtain that
	\begin{align} \label{eq_sec_stat_learn_7} 
		D(\hat{\mu}_n||\mu) \leq \frac{\log e}{\mathbf{m}_\mu} \cdot V(\hat{\mu}_n,\mu), \text{ and }
		\left|H((\hat{\mu}_n)- H(\mu_n)\right|  \leq \left[ \mathbf{M}_\mu+\frac{\log e}{\mathbf{m}_\mu}\right] \cdot 	V(\hat{\mu}_n,\mu).
	\end{align}
	Hence,  (\ref{eq_sec_stat_learn_6}) and (\ref{eq_sec_stat_learn_6b}) derive from (\ref{eq_sec_stat_learn_4}). 
	
	For the direct I-divergence, let us consider a sequence $(x_i)_{i\geq 1}$ and the following 
	function (a stopping time):
	\begin{align} \label{eq_sec_stat_learn_10} 
		T_o(x_1,x_2,\ldots) \equiv \inf \left\{n\geq 1: A_{\hat{\mu}_n(x^n_1)}=A_\mu \right\},
	\end{align}
	i.e, $T_o(x_1,x_2,\ldots)$ is the point where the support of $\hat{\mu}_n(x^n_1)$ is equal to $A_\mu$
	and, consequently, the direct I-divergence is well-defined (since $\mu \in \mathcal{F}(\mathbb{X})$).
	\footnote{By the uniform convergence of $\hat{\mu}_n$ to $\mu_n$ ($\mathbb{P}_{\mu}$-a.s.) 
	and the finite support assumption of $\mu$,  it  is  simple to verify that  
	$\mathbb{P}_{\mu}(T_o(X_1,X_2,\ldots) < \infty)=1$.} 
	Let us define the event: 
	\begin{align} \label{eq_proof_lm_con_inq_tv_0}
	\mathcal{B}_n \equiv \left\{x_1,x_2,..:T_o(x_1,x_2,..)\leq n\right\} \subset \mathbb{X}^{\mathbb{N}},
	\end{align}
 	i.e.,  the collection of sequences in $\mathbb{X}^{\mathbb{N}}$ where at time $n$,
	$A_{\hat{\mu}_n}=A_{\mu}$ and, consequently, $D(\mu||\hat{\mu}_n) < \infty$. 
	Restricted to this set, 
	\begin{align} \label{eq_proof_lm_con_inq_tv_1}
		D(\mu || \hat{\mu}_n) &\leq 
				\sum_{x \in A_{\hat{\mu}_n||\mu}} f_{\hat{\mu}_n}(x) \log \frac{f_{\hat{\mu}_n}(x)}{f_{{\mu}}(x)}
				+ \sum_{x \in A_\mu \setminus A_{\hat{\mu}_n || \mu}} f_{\hat{\mu}_n}(x) \log \frac{f_\mu(x)} {f_{\hat{\mu}_n}(x)} \\
				&\leq \log e \cdot  \sum_{x \in A_{\hat{\mu}_n||\mu}}  f_{\hat{\mu}_n}(x) \cdot \left(\frac{f_{\hat{\mu}_n}(x)}{f_{{\mu}}(x)} -1\right)\nonumber\\ 
				 &+ \log e \cdot \left[\mu(A_\mu \setminus A_{\hat{\mu}_n || \mu}) - \hat{\mu}_n((A_\mu \setminus A_{\hat{\mu}_n || \mu}))\right]\\
				\label{eq_proof_lm_con_inq_tv_1b}
				&\leq \log e  \cdot \left({1}/{\mathbf{m}_u}+1\right) V(\mu, \hat{\mu}_n),
	\end{align}
	where in the first inequality  $A_{\hat{\mu}_n||\mu}\equiv {\left\{x \in A_{\hat{\mu}_n}:  f_{\hat{\mu}_n}(x) > f_{{\mu}}(x) \right\}}$, 
	and the last is obtained by the definition of the total variational distance. 
	In addition,  let us  define the $\epsilon$-deviation set 
	$\mathcal{A}^{\epsilon}_n \equiv \left\{x_1,x_2,...:D(\mu||\hat{\mu}_n(x^n_1))>\epsilon\right\} \subset
	 \mathbb{X}^{\mathbb{N}}$. Then by additivity and monotonicity of $\mathbb{P}_\mu$,  we have that 
	\begin{align} \label{eq_proof_lm_con_inq_tv_2} 
		\mathbb{P}_\mu(\mathcal{A}^{\epsilon}_n) \leq  \mathbb{P}_\mu(\mathcal{A}^{\epsilon}_n \cap \mathcal{B}_n) + \mathbb{P}_\mu (\mathcal{B}_n^c).
	\end{align}
	By definition of $\mathcal{B}_n$,  (\ref{eq_proof_lm_con_inq_tv_1b})
	and (\ref{eq_sec_stat_learn_4})  
	\begin{align} \label{eq_proof_lm_con_inq_tv_3} 
		\mathbb{P}_\mu(\mathcal{A}^{\epsilon}_n \cap \mathcal{B}_n) &\leq 
			\mathbb{P}_\mu (V(\mu||\hat{\mu}_n) \log e  \cdot \left({1}/{\mathbf{m}_u}+1\right) >\epsilon)\nonumber\\
			& \leq 2^{\left|A_\mu\right| +1} \cdot e^{-\frac{2 n \epsilon^2}{ \log e^2  \cdot \left({1}/{\mathbf{m}_u}+1\right)^2}}.
	\end{align}
	On the other hand,  $\forall \epsilon_o \in (0,\mathbf{m}_\mu)$
	if $V(\mu, \hat{\mu}_n) \leq \epsilon_o$ then $T_o\leq n$. 
	Consequently $\mathcal{B}^c_n \subset \left\{x_1,x_2,..:V(\mu, \hat{\mu}_n(x^n_1))>\epsilon_o\right\}$, 
	and again from (\ref{eq_sec_stat_learn_4})  
	\begin{align} \label{eq_proof_lm_con_inq_tv_4} 
		\mathbb{P}_\mu(\mathcal{B}_n^c) & \leq 2^{\left|A_\mu\right|+1} \cdot e^{-2 n \epsilon_o^2}, 
	\end{align}
	for all $n\geq 1$ and $\forall \epsilon_o\in (0,\mathbf{m}_\mu)$.
	Integrating the results in  (\ref{eq_proof_lm_con_inq_tv_3}) and (\ref{eq_proof_lm_con_inq_tv_4})
	and considering $\epsilon_0 =\mathbf{m}_\mu/\sqrt{2}$,  suffices to show  (\ref{eq_sec_stat_learn_12b}).  
\end{proof}

\subsection{Theorem \ref{pro_barron_conv}}
\begin{proof}
As $(a_n)$ is $o(1)$,  it is simple to verify that 
$\lim_{n \rightarrow \infty}V(\tilde{\mu}_n, \mu)=0$, $\mathbb{P}_{\mu}$-a.s.
Also note that the support disagreement between $\tilde{\mu}_n$ and $\mu$
is bounded by the hypothesis, then 
\begin{align} \label{eq_pro_barron_conv_1} 
	&\lim_{n \rightarrow \infty}  \tilde{\mu}_n \left(A_{\mu_n} \setminus A_{\mu}\right) \cdot \log \left|A_{\tilde{\mu}_n} \setminus A_{\mu}\right| 
	\leq 
	\lim_{n \rightarrow \infty}  \tilde{\mu}_n \left(A_{\mu_n} \setminus A_{\mu}\right) \cdot \log \left|A_{v} \right|=0, 
	\ \mathbb{P}_{\mu}-a.s.
\end{align}
Therefore from Lemma \ref{fs_main_th},  
we have the strong consistency of  $H(\tilde{\mu}_n)$ and the almost sure convergence of $D( \mu||\tilde{\mu}_n)$ to zero. Note that $D( \mu||\tilde{\mu}_n)$
is uniformly upper bounded by $\log e \cdot (1/\mathbf{m}_\mu+1) V(\mu, \tilde{\mu}_n)$ 
(see (\ref{eq_proof_lm_con_inq_tv_1b}) in the proof of Theorem \ref{th_con_inq}). 
Then the convergence in probability of $D( \mu||\tilde{\mu}_n)$ implies the 
convergence of its mean \citep{breiman_1968},
which concludes the proof of the first part.  

Concerning rate of convergences, we use the following: 
\begin{align} \label{eq_pro_barron_conv_2} 
		H(\mu )- H(\tilde{\mu}_n) &= \sum_{x \in A_\mu \cap A_{\tilde{\mu}_n}} \left[f_{\tilde{\mu}_n}(x)-f_\mu(x)\right]  \log f_\mu(x)
	 + \sum_{x\in A_\mu \cap A_{\tilde{\mu}_n}} f_{\tilde{\mu}_n}(x) \log \frac{f_{\tilde{\mu}_n}(x)}{f_{\mu}(x)}\nonumber\\
	 &-\sum_{x\in A_{\tilde{\mu}_n} \setminus A_{\mu}} f_{\tilde{\mu}_n}(x) \log \frac{1}{f_{\tilde{\mu}_n}(x)}.
\end{align}
The absolute value of the first term in the right hand side (RHS) of (\ref{eq_pro_barron_conv_2}) is 
bounded by $\mathbf{M}_{\mu} \cdot V(\tilde{\mu}_n,\mu)$ and the second term is bounded 
by $\log e/{\mathbf{m}_\mu} \cdot V(\tilde{\mu}_n,\mu)$, from the assumption that $\mu \in \mathcal{F}(\mathbb{X})$. For the last term, note that $f_{\tilde{\mu}_n}(x)=a_n\cdot v(\left\{ x\right\})$ for all 
$x \in A_{\tilde{\mu}_n} \setminus A_{\mu}$ and that $A_{\tilde{\mu}_n} =A_v$, then 
$$0\leq \sum_{x\in A_{\tilde{\mu}_n} \setminus A_{\mu}} f_{\tilde{\mu}_n}(x) \log \frac{1}{f_{\tilde{\mu}_n}(x)} \leq a_n\cdot (H(v) + \log \frac{1}{a_n}\cdot v(A_v\setminus A_\mu)).$$
On the other hand, 
\begin{align} 
V(\tilde{\mu}_n,\mu)&=\frac{1}{2}\sum_{x\in A_\mu}  \left| (1-a_n)\hat{\mu}_n(\left\{ x\right\}) +a_n v(\left\{ x\right\})  - \mu(\left\{ x\right\})\right| + \sum_{x \in A_{v} \setminus A_{\mu}} a_n v(\left\{ x\right\}).\nonumber\\
	&\leq (1-a_n)\cdot V(\hat{\mu}_n, \mu) +a_n.\nonumber
\end{align}
Integrating these bounds in (\ref{eq_pro_barron_conv_2}), 
\begin{align} \label{eq_pro_barron_conv_4} 
	 \left|  H(\mu )- H(\tilde{\mu}_n)\right|  &\leq (\mathbf{M}_{\mu} + \log e/\mathbf{m}_{\mu}) \cdot ((1-a_n)\cdot V(\hat{\mu}_n, \mu) +a_n) + a_n \cdot H(v) + a_n\cdot \log \frac{1}{a_n}	\nonumber\\ 		
	 &= K_1\cdot V(\hat{\mu}_n, \mu) + K_2\cdot a_n + a_n\cdot \log \frac{1}{a_n},  
\end{align}
for  constants $K_1>0$ and $K_2>0$ function of $\mu$ and $v$.

Under the assumption that $\mu \in \mathcal{F}(\mathbb{X})$,  the {\em Hoeffding's inequality} \citep{devroye_2001,devroye1996} tells us that 
$\mathbb{P}_\mu(V(\hat{\mu}_n, \mu) >\epsilon) \leq C_1 \cdot e^{-C_2 n \epsilon^2}$
(for some distribution free constants $C_1>0$ and $C_2>0$). From this inequality,  
 $V(\hat{\mu}_n, \mu)$ goes to zero as $o(n^{-\tau})$ $\mathbb{P}_\mu$-a.s. $\forall \tau\in (0,1/2)$ and 
 $\mathbb{E}_{\mathbb{P}_\mu}(V(\hat{\mu}_n, \mu))$ is $O(1/\sqrt{n})$. On the other hand,  
under the assumption in ii) $(K_2\cdot a_n + a_n\cdot \log \frac{1}{a_n})$ is $O(1/\sqrt{n})$, which from (\ref{eq_pro_barron_conv_4}) proves the convergence rate results for $ \left|  H(\mu )- H(\tilde{\mu}_n)\right|$.

Considering the direct I-divergence, $D(\mu|| \tilde{\mu}_n) \leq \log e \cdot \sum_{x \in A_\mu}   f_\mu(x) \left|\frac{f_\mu(x)}{f_{\tilde{\mu}_n}(x)}-1\right| \leq \frac{\log e}{\mathbf{m}_{\tilde{\mu}_n}} \cdot V(\tilde{\mu}_n,\mu)$,  then the uniform convergence of  $\tilde{\mu}_n(\left\{x\right\})$ to $\mu(\left\{x\right\})$ $\mathbb{P}_\mu$-a.s. in $A_\mu$, and the fact that $\left|A_\mu\right|<\infty$, imply that  for an arbitrary small 
$\epsilon>0$ (in particular smaller than $\mathbf{m}_\mu$),  
\begin{align}\label{eq_pro_barron_conv_5} 
	\lim_{n \longrightarrow \infty} D(\mu|| \tilde{\mu}_n) &\leq \frac{\log e}{\mathbf{m}_\mu-\epsilon} \cdot \lim_{n \longrightarrow \infty}V(\tilde{\mu}_n,\mu), \ \mathbb{P}_\mu-a.s., 
\end{align}
which suffices to obtain the convergence results for the I-divergence. 
\end{proof}

\subsection{Theorem \ref{th_barron}}
\begin{proof}
	Let us define the {\em oracle Barron measure} $\tilde{\mu}_n$ by:
	\begin{equation}\label{eq_proof_th_barron_1}
		f_{\tilde{\mu}_n}(x) = \frac{d \tilde{\mu}_n}{d \lambda}(x) = f_v(x) \left[(1-a_n) \cdot  \frac{{\mu}(A_n(x))}{v(A_n(x))} + a_n\right],
	\end{equation}
	where we consider the true measure instead of its empirical version in (\ref{eq_subsec:barron_2}). 
	Then, the following convergence results can be obtained (see Proposition \ref{pro_est_error_barron} in Appendix \ref{appendix_3}), 
	\begin{align} \label{eq_proof_th_barron_9b_v2}
	\lim_{n \rightarrow \infty} \sup_{x\in A_{\tilde{\mu}_n}}\left|\frac{d \tilde{\mu}_n}{d \mu^*_n}(x)-1\right|=0, \  		\mathbb{P}_\mu-a.s.
	\end{align}
	Let $\mathcal{A}$ denote the collection   of sequences $x_1,x_2,....$ where the convergence in 
	(\ref{eq_proof_th_barron_9b_v2}) is holding\footnote{This set is almost-surely typical meaning that $\mathbb{P}_\mu(\mathcal{A})=1$.}.
	The rest of the proof reduces to show that for any arbitrary $(x_n)_{n\geq 1} \in \mathcal{A}$, 
	its respective sequence of induced measures  $\left\{ \mu^*_n: n\geq 1 \right\}$\footnote{For simplicity, the
	 dependency of $\mu^*_n$  on the sequence $(x_n)_{n\geq 1}$ will be considered implicit for the rest of the
	 proof.} satisfies the sufficient conditions  of Lemma \ref{th_con_infty_2} for entropy convergence.
	 
	 Let fix an arbitrary $(x_n)_{n\geq 1}\in \mathcal{A}$:\\
	{\bf Weak convergence $\mu^*_n \Rightarrow \mu$:}
	Without  loss of generality  we consider that $A_{\tilde{\mu}_n}=A_v$ for all $n\geq 1$. 
	Since $a_n \rightarrow 0$ and $h_n \rightarrow 0$,
	$f_{\tilde{\mu}_n}(x) \rightarrow \mu(\left\{x\right\})$ $\forall x \in A_v$,  
	we got the weak 
	convergence of $\tilde{\mu}_n$ to $\mu$.  On the other hand by definition of  $\mathcal{A}$,  
	$ \lim_{n \rightarrow \infty}  \sup_{x\in A_{\tilde{\mu}_n}} \left|\frac{f_{\tilde{\mu}_n}(x)}{f_{\mu^*_n}(x)}-1\right| =0$
	that implies that $ \lim_{n \rightarrow \infty} \left| f_{\mu^*_n}(x) - f_{\tilde{\mu}_n}(x)\right|=0$
	for all $x \in A_v$, and consequently  $\mu^*_n \Rightarrow \mu$. \\
	{\bf The condition in (\ref{eq_sec_infty_sp_2}):}
	By construction $\mu \ll {\mu}^*_n$,  $\mu \ll \tilde{\mu}_n$ and $\tilde{\mu}_n \approx {\mu}^*_n$  for all $n$, 
	then we will use the following equality: 
	\begin{equation}\label{eq_proof_th_barron_5}
		\log \frac{d \mu }{d \mu^*_n} (x) = \log \frac{d \mu }{d \tilde{\mu}_n}(x)  + \log \frac{d \tilde{\mu}_n}{d \mu^*_n }(x),
	\end{equation}
	for all $x \in A_\mu$. Concerning the approximation error term of (\ref{eq_proof_th_barron_5}), i.e.,  $\log \frac{d \mu }{d \tilde{\mu}_n}(x)$, 
	 $\forall x\in A_\mu$, 
	\begin{equation}\label{eq_proof_th_barron_2}
		\frac{d \tilde{\mu}_n}{d \mu}(x)= (1-a_n) \left[\frac{\mu(A_n(x))}{\mu(\left\{x\right\})} \frac{v(\left\{x\right\})}	{v(A_n(x))}\right] + a_n \frac{v(\left\{x\right\})}{\mu (\left\{x\right\})}.
	\end{equation}
	Given that $\mu \in \mathbb{H}(\mathbb{X}|v)$ this is equivalent to state that 
	 $\log (\frac{d \mu}{d v}(x))$ is bounded  $\mu$-almost everywhere, 
	which is equivalent to say that  $\mathbf{m}\equiv \inf_{x\in A_\mu} \frac{d \mu}{d v}(x) >0$
	and $M \equiv \sup_{x\in A_\mu} \frac{d \mu}{d v}(x) < \infty$. From this,  
	$\forall A \subset A_\mu$, 
	\begin{equation}\label{eq_proof_th_barron_3}
	\mathbf{m}   v(A) \leq \mu(A)\leq M  v(A).
	\end{equation}
	Then we have that, $\forall x \in A_\mu$
		$\frac{\mathbf{m}}{M}\leq \left[\frac{\mu(A_n(x))}{\mu(\left\{x\right\})} \frac{v(\left\{x\right\})}{v(A_n(x))}\right] \leq \frac{M}{\mathbf{m}}$.
	Therefore  for $n$ sufficient large,   $0<\frac{1}{2} \frac{\mathbf{m}}{M} \leq \frac{d \tilde{\mu}_n}{d \mu}(x) \leq \frac{M}{\mathbf{m}} + M < \infty$ for all $x$ in  $A_\mu$. Hence,  there exists $N_o>0$
	such that $\sup_{n\geq N_o} \sup_{x \in A_\mu}  \left|\log \frac{d \tilde{\mu}_n}{d \mu}(x)\right| < \infty$.\\
	For the estimation error term of (\ref{eq_proof_th_barron_5}), i.e., $\log \frac{d \tilde{\mu}_n}{d \mu^*_n }(x)$,  note that from the fact that  $(x_n)\in \mathcal{A}$, 
	and the convergence in (\ref{eq_proof_th_barron_9b_v2}),  there exists $N_1>0$ such that for all $n\geq N_1$   $	\sup_{x\in A_\mu} \left|\log \frac{d \tilde{\mu}_n}{d \mu^*_n}(x)\right| < \infty$,  given that $A_\mu\subset 	A_{\tilde{\mu}_n}=A_v$. Then using (\ref{eq_proof_th_barron_5}), 
	for all $n \geq \max \left\{N_0,N_1 \right\}$   $\sup_{x\in A_\mu} \left| \log \frac{d \mu^*_n}{d \mu}(x)\right| < \infty$, 
	which verifies (\ref{eq_sec_infty_sp_2}).\\
	%
	{\bf The condition in (\ref{eq_sec_infty_sp_3}):}
	Defining the function $\phi^*_n(x) \equiv 1_{A_v\setminus A_\mu}(x) \cdot f_{{\mu}^*_n}(x) \log (1/f_{\mu^*_n}(x))$, we want to verify that  $\lim_{n \rightarrow \infty} \int_{\mathbb{X}} \phi^*_n(x) d \lambda(x)=0$. 
	Considering that $(x_n)\in \mathcal{A}$, for all $\epsilon>0$ there exists $N(\epsilon)>0$, 
	such that $\sup_{x\in A_{\tilde{\mu}_n}} \left|\frac{f_{\tilde{\mu}_n}(x)}{f_{\mu^*_n}(x)}-1\right|<\epsilon$, then
	\begin{equation}\label{th_barron_c3_1}
		(1-\epsilon) f_{\tilde{\mu}_n}(x) < f_{\mu^*_n}(x) < (1+ \epsilon) f_{\tilde{\mu}_n}(x), \text{ for all } x \in A_v.
	\end{equation}
	From (\ref{th_barron_c3_1}),  $0 \leq \phi^*_n(x) \leq (1+ \epsilon) f_{\tilde{\mu}_n}(x) \log (1/(1- \epsilon) f_{\tilde{\mu}_n}(x))$ for all $n \geq N(\epsilon)$.  Analyzing $f_{\tilde{\mu}_n}(x)$ in (\ref{eq_proof_th_barron_1}) 	there are two scenarios: if $A_n(x)\cap A_\mu = \emptyset$   then $f_{\tilde{\mu}_n}(x)=a_n f_v(x)$, and otherwise $f_{\tilde{\mu}_n}(x)=f_v(x)(a_n+(1-a_n)\mu(A_n(x)\cap A_\mu)/v(A_n(x)))$.  Let us define: 
	\begin{align}\label{th_barron_c3_2}
		\mathcal{B}_n \equiv  \left\{x \in A_v\setminus A_\mu: A_n(x)\cap A_\mu = \emptyset \right\} \text{ and }
		\mathcal{C}_n \equiv  \left\{x\in A_v\setminus A_\mu: A_n(x)\cap A_\mu \ne \emptyset \right\}.
	\end{align}
	Then for all $n\geq N(\epsilon)$, 
	\begin{align}\label{th_barron_c3_3}
		 \sum_{\mathbb{X}} \phi^*_n(x) 
		 &\leq \sum_{A_v\setminus A_\mu} (1+ \epsilon) f_{\tilde{\mu}_n}(x) \log 1/((1- \epsilon) f_{\tilde{\mu}_n}(x)) 
		 \nonumber\\
		 &= \sum_{\mathcal{B}_n} (1+\epsilon) a_n f_v(x) \log \frac{1}{(1-\epsilon)a_n f_v(x)} + \sum_{\mathbb{X}} \tilde{\phi}_n(x), 
	\end{align}
	with $\tilde{\phi}_n(x) \equiv 1_{\mathcal{C}_n}(x) \cdot (1+\epsilon) f_{\tilde{\mu}_n}(x)  \log \frac{1}{(1-	\epsilon)f_{\tilde{\mu}_n}(x)}$.  The left term in (\ref{th_barron_c3_3}) is upper bounded by 
	$a_n(1+\epsilon) (H(v) + \log (1/a_n))$ which goes to zero with $n$ from $(a_n)$ being $o(1)$ 
	and the fact that $v\in \mathbb{H}(\mathbb{X})$.  For the right term in (\ref{th_barron_c3_3}), $(h_n)$ being $o(1)$ implies that  $\forall x\in A_v\setminus A_\mu$,  $x$ belongs to $\mathcal{B}_n$ eventually in $n$, then  $\tilde{\phi}_n(x)$ 
	tends to zero point-wise as $n$ goes to infinity. On the other hand, for all $x \in \mathcal{C}_n$ 
	(see 	(\ref{th_barron_c3_2})),  we have that
	\begin{align}\label{th_barron_c3_4}
		\frac{1}{1/m + 1}\leq \frac{\mu(A_n(x) \cap A_\mu )}{v(A_n(x) \cap A_\mu) +  v(A_v\setminus A_\mu)} \leq 	\frac{\mu(A_n(x))}{v(A_n(x))} \leq \frac{\mu(A_n(x)\cap A_\mu)}{v(A_n(x)\cap A_\mu)} \leq M.
	\end{align}
	These inequalities derive from (\ref{eq_proof_th_barron_3}). Consequently for all $x\in \mathbb{X}$,  
	if  $n$ sufficiently large such that $a_n<0.5$, then 
	\begin{align}\label{th_barron_c3_5}
	0\leq \tilde{\phi}_n(x) &\leq (1+\epsilon)(a_n+(1-a_n)M)f_v(x) \log \frac{1}{(1-\epsilon)(a_n+(1-a_n)m/(m+1))} 	\nonumber\\
				& \leq (1 + \epsilon)(1+M)f_v(x)   \left[ \log \frac{2(m+1)}{(1-\epsilon)} + \log \frac{1}{f_v(x)} \right] .
	\end{align}
	Hence from (\ref{th_barron_c3_2}), $\tilde{\phi}_n(x)$ is  bounded by a  fix function that is $l_1(\mathbb{X})$ by the assumption that $v\in \mathbb{H}(\mathbb{X})$. Then by the {\em dominated convergence theorem} \citep{varadhan_2001} and (\ref{th_barron_c3_3}),   
	$$\lim_{n \rightarrow \infty} \sum_{\mathbb{X}} \phi^*_n(x) 
	\leq  \lim_{n \rightarrow \infty} \sum_{\mathbb{X}} \tilde{\phi}_n(x). 
	$$ 
	
	In summary, we have shown that for any arbitrary  $(x_n)\in \mathcal{A}$ the sufficient conditions of Lemma \ref{th_con_infty_2} are satisfied, which proves the result in (\ref{eq_subsec:barron_3}) reminding that $\mathbb{P}_\mu(\mathcal{A})=1$. 
\end{proof}

\subsection{Theorem \ref{th_data_driven}}
\begin{proof}
	Let us  first introduce the oracle measure 
	\begin{equation} \label{eq_data_driven_3b}
	\mu_{\epsilon_n} \equiv  \mu(\cdot |\Gamma_{\epsilon_n}) \in \mathcal{P}(\mathbb{X}).
	\end{equation}
	Note that $\mu_{\epsilon_n}$ is a random probability measure (function of the i.i.d sequence $X_1,..,X_n$) as $\Gamma_{\epsilon_n}$ is a data-driven set (\ref{eq_data_driven_1}).  We will first show that: 
	\begin{equation} \label{eq_data_driven_3c}
	\lim_{n \rightarrow \infty}  H(\mu_{\epsilon_n})=H(\mu) 
	\text{ and } \lim_{n \rightarrow \infty}  D(\mu_{\epsilon_n}|| \mu)=0, \  \mathbb{P}_\mu-a.s. 
	\end{equation}
	Under the assumption on $(\epsilon_n)$ of Theorem \ref{th_data_driven},    
	$\lim_{n \rightarrow \infty}  \left| \mu(\Gamma_{\epsilon_n})- \hat{\mu}_n(\Gamma_{\epsilon_n}) \right|=0$,  
	$\mathbb{P}_\mu$-a.s.\footnote{This result derives from the fact that $\lim_{n \rightarrow \infty} V(\mu/\sigma_{\epsilon_n}, \hat{\mu}_n/\sigma_{\epsilon_n}) =0$,  $\mathbb{P}_\mu\text{-}a.s.$ , from (\ref{eq_data_driven_new8}).}  
	In addition,  since $(\epsilon_n)$ is $o(1)$ then $\lim_{n \rightarrow \infty} \hat{\mu}_n(\Gamma_{\epsilon_n}) =1$,  
	which implies that $\lim_{n \rightarrow \infty}   \mu(\Gamma_{\epsilon_n})=1$ $\mathbb{P}_\mu$-a.s.  
	From this $\mu_{\epsilon_n} \Rightarrow \mu$, $\mathbb{P}_\mu$-a.s.  
	Let us  consider a sequences $(x_n)$ where 
	$\lim_{n \rightarrow \infty}   \mu(\Gamma_{\epsilon_n})=1$.  Constrained to that 
	\begin{equation}
		\lim\sup_{n \rightarrow \infty} \sup_{x\in A_\mu} \frac{f_{\mu_{\epsilon_n}}(x) }{f_{\mu}(x)} = \lim \sup_{n 			\rightarrow \infty} \frac{1}{\mu(\Gamma_{\epsilon_n})} < \infty,
	\end{equation}
	then there is $N>0$ such that $\sup_{n>N} \sup_{x\in A_\mu} \frac{f_{\mu_{\epsilon_n}}(x) }{f_{\mu}(x)} <\infty$. 
	Hence from Lemma  \ref{th_con_infty_1},   $\lim_{n \rightarrow \infty} D(\mu_{\epsilon_n} || \mu)=0$ and 
	$\lim_{n \rightarrow \infty}  \left| H(\mu_{\epsilon_n}) - H(\mu) \right|=0$.  Finally,  the set of sequences $(x_n)$ 
	where $\lim_{n \rightarrow \infty}   \mu(\Gamma_{\epsilon_n})=1$ has probability one with respect 
	to $\mathbb{P}_\mu$, which proves (\ref{eq_data_driven_3c}).
	
	For the rest of the proof,  we concentrate on the analysis of  
	$\left|H(\hat{\mu}_{n,\epsilon_n})- H(\mu_{\epsilon_n}) \right|$ 
	that can be attributed to the estimation error aspect of the problem. 
	It is worth noting that by construction  $supp(\hat{\mu}_{n,\epsilon_n})=supp(\mu_{\epsilon_n})=\Gamma_{\epsilon_n}$, $\mathbb{P}_\mu$-a.s., consequently we can use
	\begin{equation} \label{eq_data_driven_4}
		H(\hat{\mu}_{n,\epsilon_n})- H(\mu_{\epsilon_n}) = \sum_{x\in \Gamma_{\epsilon_n}}  \left[ \mu_{\epsilon_n} (\left\{x \right\})- \hat{\mu}_{n,\epsilon_n} (\left\{x \right\}) \right] \log \hat{\mu}_{n,\epsilon_n} (\left\{ x \right\}) + D(\mu_{\epsilon_n} || \hat{\mu}_{n,\epsilon_n}).
	\end{equation}
		
	The first term on the RHS  of (\ref{eq_data_driven_4}) is upper bounded  by 
	$\log 1/ \mathbf{m}_{\hat{\mu}_n}^{\epsilon_n} \cdot V(\mu_{\epsilon_n}, \hat{\mu}_{n,\epsilon_n})$ $\leq \log 1/\epsilon_n \cdot V(\mu_{\epsilon_n}, \hat{\mu}_{n,\epsilon_n})$. Concerning the second term on the RHS  of (\ref{eq_data_driven_4}), it is possible to show  (details presented in Appendix \ref{appendix_div_bound_ddp}) that
	\begin{equation} \label{eq_data_driven_new5}
		D(\mu_{\epsilon_n} || \hat{\mu}_{n,\epsilon_n}) \leq  \frac{2 \log \frac{e}{\epsilon_n}}{\mu(\Gamma_{\epsilon_n})} \cdot V(\mu/\sigma_{\epsilon_n}, \hat{\mu}_n/\sigma_{\epsilon_n}),
	\end{equation}
	where 
	\begin{equation} \label{eq_data_driven_new5b}
	V(\mu/\sigma_{\epsilon_n}, \hat{\mu}_n/\sigma_{\epsilon_n}) \equiv \sup_{A \in \sigma_{\epsilon_n}}   \left| \mu(A) - \hat{\mu}_n(A)\right|.
	\end{equation}
	In addition, it can be verified  (details presented in Appendix \ref{appendix_1}) that
	\begin{equation} \label{eq_data_driven_5}
		V(\mu_{\epsilon_n}, \hat{\mu}_{n,\epsilon_n}) \leq K \cdot  V(\mu/\sigma_{\epsilon_n}, \hat{\mu}_n/\sigma_{\epsilon_n}), 
	\end{equation}
	for some universal constant $K>0$.  Therefore from (\ref{eq_data_driven_4}), (\ref{eq_data_driven_new5}) and (\ref{eq_data_driven_5}), there is $C>0$ such that:
	\begin{equation} \label{eq_data_driven_new6}
		\left|  H(\hat{\mu}_{n,\epsilon_n})- H(\mu_{\epsilon_n})\right|  \leq  \frac{C}{\mu(\Gamma_{\epsilon_n})} \log \frac{1}{\epsilon_n} \cdot  V(\mu/\sigma_{\epsilon_n}, \hat{\mu}_n/\sigma_{\epsilon_n}).
	\end{equation}
	As mentioned before,  $\mu(\Gamma_{\epsilon_n})$ goes to $1$ almost surely,  then we need to concentrate 
	on the analysis of the asymptotic behavior of $\log1/\epsilon_n  \cdot V(\mu/\sigma_{\epsilon_n}, \hat{\mu}_n/\sigma_{\epsilon_n})$.
	From {\em Hoeffding's inequality} \citep{devroye_2001}, we have that $\forall \delta>0$
	\begin{equation} \label{eq_data_driven_new7}
		\mathbb{P}^n_{\mu} \left(  \log1/\epsilon_n  \cdot V(\mu/\sigma_{\epsilon_n}, \hat{\mu}_n/\sigma_{\epsilon_n})  >\delta  \right) \leq 2^{\left| \Gamma_{\epsilon_n}\right|+1} \cdot e^{-\frac{ 2n \delta^2}{(\log 1/\epsilon_n)^2}}, 
	\end{equation}
	considering that by construction $\left| \sigma_{\epsilon_n} \right| \leq 2^{\left| \Gamma_{\epsilon_n}\right|+1} \leq  2^{1/\epsilon_n +1}$.
	Assuming that $(\epsilon_n)$ is $O(n^{-\tau})$, 
	$$\ln \mathbb{P}^n_{\mu} \left( \log 1/\epsilon_n  \cdot V(\mu/\sigma_{\epsilon_n}, \hat{\mu}_n/\sigma_{\epsilon_n})  >\delta  \right) \leq  (n^\tau+1) \ln 2 -\frac{2n\delta^2}{\tau \log n}.$$ 
	Therefore for all $\tau\in (0,1)$, $\delta>0$ and any arbitrary $l\in (\tau, 1)$
	\begin{equation} \label{eq_data_driven_new8}
		\lim\sup_{n \rightarrow \infty}\frac{1}{n^l} \cdot \ln \mathbb{P}^n_{\mu} \left( \log 1/\epsilon_n  \cdot V(\mu/\sigma_{\epsilon_n}, \hat{\mu}_n/\sigma_{\epsilon_n})  >\delta  \right) < 0. 
	\end{equation}
	This is sufficient to  show that $\sum_{n\geq 1} \mathbb{P}^n_{\mu} \left(  \log1/\epsilon_n  \cdot V(\mu/\sigma_{\epsilon_n}, \hat{\mu}_n/\sigma_{\epsilon_n})  >\delta  \right) < \infty$ that concludes the argument 
	from the  {\em Borel-Cantelli Lemma}.

\end{proof}

\subsection{Theorem \ref{th_rate_data_driven_power}}
\label{proof_th_rate_data_driven_power}
\begin{proof}
	We  consider 
	\begin{equation}\label{pr_th_rate_data_driven_0}
	\left| H(\mu) - H(\hat{\mu}_{n,\epsilon_n}) \right| \leq  \left| H(\mu) - H(\mu_{\epsilon_n}) \right| +  \left| H(\mu_{\epsilon_n}) - H(\hat{\mu}_{n,\epsilon_n}) \right| 
	\end{equation}
	to analize the approximation  and the estimation error terms separately. 
	
\subsubsection{Approximation error analysis:} 
	 Note that 
	 $\left| H(\mu) - H(\mu_{\epsilon_n}) \right|$  is a random object
	  as $\mu_{\epsilon_n}$ in (\ref{eq_data_driven_3b}) is a function  of the data-dependent partition 
	  and, consequently, a function of $X_1,..,X_n$. In the  following,  
	  we consider the oracle set
	  \begin{equation}\label{pr_th_rate_data_driven_1}
	  	\tilde{\Gamma}_{\epsilon_n} \equiv \left\{x \in \mathbb{X}:  \mu( \left\{x \right\})\geq {\epsilon_n} \right\},  
	   \end{equation}
	   and the oracle conditional measure\footnote{$\tilde{\Gamma}_{\epsilon_n}$ is a deterministic function of $(\epsilon_n)$ and so is  the measure $\tilde{\mu}_{\epsilon_n}$ in (\ref{pr_th_rate_data_driven_2}).}
	  \begin{equation}\label{pr_th_rate_data_driven_2}   
	  	\tilde{\mu}_{\epsilon_n} \equiv  \mu(\cdot | \tilde{\Gamma}_{\epsilon_n}) \in \mathcal{P}(\mathbb{X}).
	   \end{equation} 
	   From  definitions and triangular inequality:
	  \begin{align}\label{pr_th_rate_data_driven_3a}   
	 		 \left| H(\mu) - H( \tilde{\mu}_{\epsilon_n}) \right|  &\leq  \sum_{x\in \tilde{\Gamma}^c_{\epsilon_n}} \mu ( \left\{x \right\}) \log \frac{1}{\mu ( \left\{x \right\}) } + \log\frac{1}{\mu (\tilde{\Gamma}_{\epsilon_n})} \nonumber\\
			 &+ \left(\frac{1}{\mu (\tilde{\Gamma}_{\epsilon_n})}  - 1 \right) \cdot   \sum_{x\in \tilde{\Gamma}_{\epsilon_n}} \mu ( \left\{ x\right\}) \log \frac{1}{\mu ( \left\{x \right\})}, 
 	   \end{align} 
	   and, similarly,  the approximation error is bounded by
	  \begin{align}\label{pr_th_rate_data_driven_3b}   
	 		 \left| H(\mu) - H(\mu_{\epsilon_n}) \right|  &\leq  \sum_{x\in \Gamma^c_{\epsilon_n}} \mu ( \left\{x \right\}) \log \frac{1}{\mu ( \left\{x \right\}) } \nonumber\\
			 &+ \log\frac{1}{\mu (\Gamma_{\epsilon_n})} +  \left(\frac{1}{\mu (\Gamma_{\epsilon_n})}  -1 \right) \cdot   \sum_{x\in \Gamma_{\epsilon_n}} \mu ( \left\{ x\right\}) \log \frac{1}{\mu ( \left\{x \right\})}. 
 	   \end{align} 
	   We denote the RHS of (\ref{pr_th_rate_data_driven_3a}) and (\ref{pr_th_rate_data_driven_3b}) by $a_{\epsilon_n}$ and $b_{\epsilon_n}(X_1,..,X_n)$, respectively. 
	   
	   We can show that if $(\epsilon_n)$ is $O(n^{-\tau})$ and $\tau \in (0,1/2)$, then 
		\begin{equation}\label{pr_th_rate_data_driven_3c}   
		\lim \sup_{n \rightarrow \infty} b_{\epsilon_n}(X_1,..,X_n) - a_{2 \epsilon_n} \leq 0, \  \mathbb{P}_{\mu}-a.s.,
		\end{equation}
	 which from (\ref{pr_th_rate_data_driven_3b})   implies that $\left| H(\mu) - H(\mu_{\epsilon_n}) \right|$ 
	   is $O(a_{2 \epsilon_n})$, $\mathbb{P}_{\mu}$-a.s.
	   The proof of (\ref{pr_th_rate_data_driven_3c}) is presented in Appendix \ref{app_prop_oracle_appro_err_apr}.
	   
	   Then, we  need to analyze the rate of convergence of the deterministic sequence $(a_{2\epsilon_n})$.
	   Analyzing the RHS of (\ref{pr_th_rate_data_driven_3a}), we recognize two independent terms: 
	   the partial entropy sum $\sum_{x\in \tilde{\Gamma}^c_{\epsilon_n}} \mu ( \left\{x \right\}) \log \frac{1}{\mu ( \left\{x \right\}) }$ and the rest that is bounded asymptotically by   $\mu(\tilde{\Gamma}^c_{\epsilon_n}) (1+H(\mu))$, using the fact that $\ln x \leq x-1$ for $x\geq 1$. Here is where the tail condition on $\mu$ plays a  role. 
	   From the tail condition, we have that 
	   \begin{align}\label{pr_th_rate_data_driven_5}   
	   	\mu(\tilde{\Gamma}^c_{\epsilon_n}) &\leq \mu  \left( \left\{ ({k_o}/{\epsilon_n})^{1/p}+1,  ({k_o}/{\epsilon_n})^{1/p}+2,  ({k_o}/{\epsilon_n})^{1/p}+3, \ldots \right\} \right)  = \sum_{x\geq  (\frac{k_o}{\epsilon_n})^{1/p}+1} \mu(\left\{x \right\}) \nonumber\\ 
		&\leq k_1 \cdot \mathcal{S}_{ (\frac{k_o}{\epsilon_n})^{1/p}+1},
	  \end{align} 
	  where $\mathcal{S}_{x_o}\equiv \sum_{x\geq x_o} x^{-p}$.  Similarly as $\left\{0,1,.., ({k_o}/{\epsilon_n})^{1/p}\right\} \subset \tilde{\Gamma}_{\epsilon_n}$, then
	   \begin{align}\label{pr_th_rate_data_driven_6}   
	   		\sum_{x\in \tilde{\Gamma}^c_{\epsilon_n}} \mu ( \left\{x \right\}) \log \frac{1}{\mu ( \left\{x \right\})} &\leq \sum_{x\geq  (\frac{k_o}{\epsilon_n})^{1/p}+1}  \mu ( \left\{x \right\}) \log \frac{1}{\mu ( \left\{x \right\})} \leq \sum_{x\geq  (\frac{k_o}{\epsilon_n})^{1/p}+1} k_1 x^{-p} \cdot \log \frac{1}{k_0 x^{-p}} \nonumber\\
		& \leq k_1\log p \cdot  \mathcal{R}_{(\frac{k_o}{\epsilon_n})^{1/p}+1} + k_1 \log 1/k_0 \cdot \mathcal{S}_{(\frac{k_o}{\epsilon_n})^{1/p}+1},
	  \end{align} 
	  where  $\mathcal{R}_{x_o}\equiv  \sum_{x\geq x_o} x^{-p} \log x$.   
	  
	  In Appendix \ref{aux_th_rate_data_driven}, 
	  it is shown that $\mathcal{S}_{x_o} \leq C_0\cdot x_o^{1-p}$ and $\mathcal{R}_{x_o} \leq C_1 \cdot x_o^{1-p}$ for constants  $C_1>0$ and $C_0>0$. Integrating these results in the RHS of (\ref{pr_th_rate_data_driven_5}) and (\ref{pr_th_rate_data_driven_6}) and considering that $(\epsilon_n)$ is $O(n^{-\tau})$, we have that 
	   both $\mu(\tilde{\Gamma}^c_{\epsilon_n})$ and $\sum_{x\in \tilde{\Gamma}^c_{\epsilon_n}} \mu ( \left\{x \right\}) \log \frac{1}{\mu ( \left\{x \right\})}$ are  $O(n^{-\frac{\tau(p-1)}{p}})$. This implies that our oracle sequence $(a_{\epsilon_n})$ is $O(n^{-\frac{\tau(p-1)}{p}})$. 
	   
	   In conclusion, if $\epsilon_n$ is $O(n^{-\tau})$ for $\tau \in (0,1/2)$, it follows that 
	   \begin{equation}\label{pr_th_rate_data_driven_7}   
	   \left| H(\mu) - H(\mu_{\epsilon_n}) \right|  \text{ is } O(n^{-\frac{\tau(p-1)}{p}}),\  \mathbb{P}_\mu-a.s.
	    \end{equation} 

\subsubsection{Estimation error analysis:} 
Let us consider $\left| H(\mu_{\epsilon_n}) - H(\hat{\mu}_{n,\epsilon_n}) \right|$, 
from the bound in (\ref{eq_data_driven_new6}) and the fact that for any $\tau\in (0,1)$, 
$\lim_{n \rightarrow \infty}{\mu(\Gamma_{\epsilon_n})}=1$ $\mathbb{P}_\mu$-a.s. from (\ref{eq_data_driven_new8}),   
the problem reduces to analyze the rate of convergence  of the following random object: 
\begin{equation}\label{pr_th_rate_data_drive_8}   
	\rho_n(X_1,..,X_n) \equiv \log \frac{1}{\epsilon_n} \cdot V(\mu/\sigma{(\Gamma_{\epsilon_n})}, \hat{\mu}_n/\sigma(\Gamma_{\epsilon_n})).
\end{equation} 
We will analize, instead, the  oracle version of $\rho_n(X_1,..,X_n)$ given by: 
\begin{equation}\label{pr_th_rate_data_drive_9}   
	\xi_n(X_1,..,X_n) \equiv \log \frac{1}{\epsilon_n} \cdot V(\mu/\sigma(\tilde{\Gamma}_{\epsilon_n/2}), \hat{\mu}_n/\sigma(\tilde{\Gamma}_{\epsilon_n/2})),
\end{equation} 
where $\tilde{\Gamma}_{\epsilon} \equiv \left\{x\in \mathbb{X}: \mu(\left\{ x\right\})\geq \epsilon \right\}$ is the oracle counterpart of ${\Gamma}_{\epsilon}$ in (\ref{eq_data_driven_1}). To do so, we can show that if $\epsilon_n$ is $O(n^{-\tau})$  with $\tau\in (0, 1/2)$, then 
\begin{equation}\label{pr_th_rate_data_drive_10}   
	\lim \inf_{n \rightarrow \infty} \xi_n(X_1,..,X_n) - \rho_n(X_1,..,X_n) \geq 0, \ \ \mathbb{P}_\mu-a.s.
\end{equation}
The proof of (\ref{pr_th_rate_data_drive_10}) is presented in Appendix \ref{app_prop_oracle_est_err_aprt}.

Moving to the almost sure rate of convergence  of $\xi_n(X_1,..,X_n)$, it is simple to show 
for our $p$-power dominating distribution that if $(\epsilon_n)$ is $O(n^{-\tau})$ and $\tau\in (0,p)$ 
then $$\lim_{n \rightarrow \infty}\xi_n(X_1,..,X_n)=0\text{ }\mathbb{P}_\mu-a.s,$$ and, more specifically,  
\begin{equation}\label{pr_th_rate_data_drive_10b}   
\xi_n(X_1,..,X_n) \text{ is } o(n^{-q}) \text{ for all } q\in (0,(1-\tau/p)/2), \ \mathbb{P}_\mu-a.s.
\end{equation}
The argument is presented in Appendix \ref{app_rate_oracle_est_error}.

In conclusion, if $\epsilon_n$ is $O(n^{-\tau})$ for $\tau \in (0,1/2)$, it follows that 
\begin{equation}\label{pr_th_rate_data_driven_10c}   
	\left| H(\mu_{\epsilon_n}) - H(\hat{\mu}_{n,\epsilon_n}) \right|  \text{ is } O(n^{-q}),\  \mathbb{P}_\mu-a.s.,
\end{equation} 
for all $q\in (0,(1-\tau/p)/2)$.

\subsubsection{Estimation vs. approximation error:} 
Coming back to (\ref{pr_th_rate_data_driven_0}) and using (\ref{pr_th_rate_data_driven_7}) and (\ref{pr_th_rate_data_driven_10c}),  the analysis reduces to finding the solution $\tau^*$  in $(0,1/2)$ that offers the 
best trade-off between the estimation and approximation error rate:
\begin{equation}\label{pr_th_rate_data_driven_13}   
	\tau^* \equiv \arg \max_{\tau (0,1/2)} \min \left\{ \frac{(1-\tau/p)}{2}, \frac{\tau(p-1)}{p}\right\}.
\end{equation}  
It is simple to verify that $\tau^*=1/2$. Then by considering $\tau$ arbitrary close to the admissible limit $1/2$,  we can 
achieve a rate of convergence for $ \left| H(\mu) - H(\hat{\mu}_{n,\epsilon_n}) \right|$
that is arbitrary close to $O(n^{-\frac{1}{2}(1-1/p)})$, $\mathbb{P}$-a.s. 

More formally,  for any $l \in (0,\frac{1}{2}(1-1/p))$ we can take $\tau\in (\frac{l}{(1-1/p)}, \frac{1}{2})$
where  $\left| H(\mu) - H(\hat{\mu}_{n,\epsilon_n}) \right|$ is $o(n^{-l})$, $\mathbb{P}_\mu$-a.s.,  
from (\ref{pr_th_rate_data_driven_7}) and (\ref{pr_th_rate_data_driven_10c}).

Finally, a simple corollary of this analysis is to consider $\tau(p)=\frac{1}{2+1/p}<1/2$ where: 
\begin{equation}\label{pr_th_rate_data_driven_14}   
	\left| H(\mu) - H(\hat{\mu}_{n,\epsilon_n}) \right| \text{ is } O(n^{- \frac{1-{1}/{p}}{2+1/p} }), \ \mathbb{P}_\mu-a.s.,   
\end{equation}  
which concludes the argument. 
\end{proof}

\subsection{Theorem \ref{th_rate_data_driven_exp_family}}
\label{proof_th_rate_data_driven_exp_family}
\begin{proof}
	The argument follows the proof of Th.~\ref{th_rate_data_driven_power}. 
	In particular, we use the estimation-approximation error bound: 
	\begin{equation} \label{eq_proof_th_rate_data_driven_exp_family_0}
	\left| H(\mu) - H(\hat{\mu}_{n,\epsilon_n}) \right| \leq  \left| H(\mu) - H(\mu_{\epsilon_n}) \right| +  \left| H(\mu_{\epsilon_n}) - H(\hat{\mu}_{n,\epsilon_n}) \right|,
	\end{equation} 
	and the following two results derived in the proof of Th.~\ref{th_rate_data_driven_power}: If $(\epsilon_n)$ is $O(n^{-\tau})$ with $\tau \in (0,1/2)$ then (for the approximation error)
	\begin{equation} \label{eq_proof_th_rate_data_driven_exp_family_1}
		\left| H(\mu) - H(\mu_{\epsilon_n}) \right| \text{ is }  O(a_{2\epsilon_{n}})\  \mathbb{P}_{\mu}-a.s.,
	\end{equation}
	with $a_{\epsilon_n} = \sum_{x\in \tilde{\Gamma}^c_{\epsilon_n}} \mu ( \left\{x \right\}) \log \frac{1}{\mu ( \left\{x \right\}) } + \mu(\tilde{\Gamma}^c_{\epsilon_n}) (1+H(\mu)),$
	while (for the estimation error)
	\begin{equation} \label{eq_proof_th_rate_data_driven_exp_family_2}
		\left| H(\mu_{\epsilon_n}) - H(\hat{\mu}_{n,\epsilon_n}) \right| \text{ is } O (\xi_n(X_1,..,X_n)) \  \mathbb{P}_{\mu}-a.s.,
	\end{equation}
	with $\xi_n(X_1,..,X_n) = \log \frac{1}{\epsilon_n} \cdot V(\mu/\sigma(\tilde{\Gamma}_{\epsilon_n/2}), \hat{\mu}_n/\sigma(\tilde{\Gamma}_{\epsilon_n/2})).$
	
	For the estimation error,  we need to bound the rate of convergence of $\xi_n(X_1,..,X_n)$ to zero almost surely.  We first note that $ \left\{1,..,x_o(\epsilon_n) \right\}=\tilde{\Gamma}_{\epsilon_n}$ with $x_o(\epsilon_n)= \lfloor  1/\alpha \ln (k_0/\epsilon_n)\rfloor$.  Then from {\em Hoeffding's inequality} we have that
	\begin{align} \label{eq_proof_th_rate_data_driven_exp_family_6}
		\mathbb{P}^n_{\mu} \left(  \left\{\xi_n(X_1,..,X_n) > \delta  \right\}  \right) &\leq 2^{\left| (\tilde{\Gamma}_{\epsilon_n/2}) \right|} \cdot  e^{-2n \frac{\delta^2}{\log(1/\epsilon_n)^2}} \nonumber\\
					&\leq 2^{1/\alpha \ln (2k_0/\epsilon_n)+1} \cdot e^{-2n \frac{\delta^2}{\log(1/\epsilon_n)^2}}.
	\end{align}
	Considering $\epsilon_{n}=O(n^{-\tau})$,  an arbitrary sequence $(\delta_n)$ being $o(1)$ and $l>0$, 
	it follows from (\ref{eq_proof_th_rate_data_driven_exp_family_6}) that 
	\begin{align} \label{eq_proof_th_rate_data_driven_exp_family_7}
		\frac{1}{n^l} \cdot \ln \mathbb{P}^n_{\mu} \left(  \left\{\xi_n(X_1,..,X_n) > \delta_n  \right\}  \right) & \leq 
				\frac{1}{n^l} \ln(2)  \left[ 1/\alpha \ln (2k_0/\epsilon_n)+1 \right] - n^{1-l}  \frac{\delta_n^2}{\log(1/\epsilon_n)^2}.  
	\end{align}
	We note that the first term in the RHS of (\ref{eq_proof_th_rate_data_driven_exp_family_7}) 
	is $O(\frac{1}{n^l}\log n)$ and goes to zero for all $l>0$, while the second term is $O(n^{1-l}\frac{\delta_n^2}{\log n^2})$.  
	If  we consider $\delta_n=O(n^{-q})$, this second term is $O(n^{1-2q-l} \cdot \frac{1}{\log n^2})$. Therefore, 
	for any $q\in (0,1/2)$ we can take an arbitrary $l\in (0,1-2q]$ such that 
	$\mathbb{P}^n_{\mu} \left(  \left\{\xi_n(X_1,..,X_n) > \delta_n  \right\}  \right)$ is $O(e^{-n^l})$ from (\ref{eq_proof_th_rate_data_driven_exp_family_7}). This result implies,  from the 
	{\em Borel-Cantelli Lemma}, that $\xi_n(X_1,..,X_n)$ is $o(\delta_n)$, $\mathbb{P}_{\mu}$-a.s, which in summary shows that $\left| H(\mu_{\epsilon_n}) - H(\hat{\mu}_{n,\epsilon_n}) \right|$ is $O(n^{-q})$ for all $q\in (0,1/2)$.
	
	For the approximation error, it is simple to verify that: 
	\begin{equation} \label{eq_proof_th_rate_data_driven_exp_family_3}
		\mu(\tilde{\Gamma}^c_{\epsilon_n}) \leq k_1 \cdot \sum_{x  \geq x_o(\epsilon_n)+1} e^{-\alpha x}  = k_1  \cdot \tilde{\mathcal{S}}_{x_o(\epsilon_n)+1} 
	\end{equation}
	and 
	\begin{align} \label{eq_proof_th_rate_data_driven_exp_family_4}
		\sum_{x\in \tilde{\Gamma}^c_{\epsilon_n}} \mu ( \left\{x \right\}) \log \frac{1}{\mu ( \left\{x \right\})} &\leq  \sum_{x  \geq x_o(\epsilon_n)+1} k_1 e^{-\alpha x} \log \frac{1}{k_0 e^{-\alpha x}} = k_1 \log \frac{1}{k_0}  \cdot \tilde{\mathcal{S}}_{x_o(\epsilon_n)+1} \nonumber\\ 
		&+ \alpha \log e \cdot k_1  \cdot \tilde{\mathcal{R}}_{x_o(\epsilon_n)+1}, 
	\end{align}
	where $\tilde{\mathcal{S}}_{x_o} \equiv \sum_{x  \geq x_o} e^{-\alpha x}$  and $\tilde{\mathcal{R}}_{x_o} \equiv \sum_{x  \geq x_o} x\cdot e^{-\alpha x}$. At this point, it is not difficult to show that $\tilde{\mathcal{S}}_{x_o} \leq M_1 e^{-\alpha x_o}$ and  
	$\tilde{\mathcal{R}}_{x_o}\leq M_2 e^{-\alpha x_o}\cdot x_o$ for some constants $M_1>0$ and $M_2>0$. Integrating these 
	partial steps,  we have that
	\begin{align} \label{eq_proof_th_rate_data_driven_exp_family_5}
		a_{\epsilon_n} \leq k_1 (1+H(\mu)+\log \frac{1}{k_0}) \cdot  \tilde{\mathcal{S}}_{x_o(\epsilon_n)+1} + \alpha \log e \cdot k_1  \cdot \tilde{\mathcal{R}}_{x_o(\epsilon_n)+1}
		    \leq O_1\cdot \epsilon_n + O_2\cdot \epsilon_n \log \frac{1}{\epsilon_n}
	\end{align}
	for some constant $O_1>0$ and $O_2>0$. 
	The last step is from the evaluation of $x_o(\epsilon_n)= \lfloor  1/\alpha \ln (k_0/\epsilon_n)\rfloor$. 
	Therefore from (\ref{eq_proof_th_rate_data_driven_exp_family_1}) and (\ref{eq_proof_th_rate_data_driven_exp_family_5}), it follows that $\left| H(\mu) - H(\mu_{\epsilon_n}) \right|$ 
	is $O(n^{-\tau} \log n)$  $\mathbb{P}_{\mu}$-a.s. for all $\tau \in (0,1/2)$. 
	
	The argument concludes by integrating in (\ref{eq_proof_th_rate_data_driven_exp_family_0}) the almost sure convergence results obtained for the estimation and approximation in errors. 
\end{proof}

\subsection{Theorem \ref{th_rate_data_driven_finite_support}}
\label{proof_th_rate_data_driven_finite_support}
\begin{proof}
	Let us define the event 
	\begin{equation}\label{eq_proof_th_rate_data_driven_finite_support_1}
		\mathcal{B}^\epsilon_n  =  \left\{x^n_1 \in \mathbb{X}^n: \Gamma_\epsilon(x^n_1)=A_\mu \right\},   
	\end{equation}
	that represents the detection of the support of $\mu$ from the data for a given $\epsilon>0$ in (\ref{eq_data_driven_1}). 
	Note that the dependency on the data for $\Gamma_\epsilon$  is made explicit in this notation. 
	In addition, let us consider the deviation event
	\begin{equation}\label{eq_proof_th_rate_data_driven_finite_support_2}
		\mathcal{A}^\epsilon_{n}(\mu)= \left\{ x^n_1 \in \mathbb{X}^n: V(\mu,\hat{\mu}_n) >\epsilon \right\}.
	\end{equation}
	By the hypothesis that $\left| A_\mu\right|<\infty$, then $\mathbf{m}_\mu=\min_{x\in A_\mu} f_\mu(x)>0$.  
	Therefore if  $x^n_1\in (\mathcal{A}^{\mathbf{m}_\mu/2}_{n}(\mu))^c$  then $\hat{\mu}_n( \left\{ x \right\}) \geq  \mathbf{m}_\mu/2$  for  all $x\in A_\mu$, which implies that  $(\mathcal{B}^\epsilon_n)^c\subset \mathcal{A}^{\mathbf{m}_\mu/2}_{n}(\mu)$ as long as $0<\epsilon\leq \mathbf{m}_\mu/2$. 
	Using the hypothesis that $\epsilon_n \rightarrow 0$,  there is $N>0$ such that for all $n\geq N$ $(\mathcal{B}^{\epsilon_n}_n)^c\subset \mathcal{A}^{\mathbf{m}_\mu/2}_{n}(\mu)$ and, consequently, 
	\begin{equation}\label{eq_proof_th_rate_data_driven_finite_support_3}
		\mathbb{P}^n _\mu ((\mathcal{B}^{\epsilon_n}_n)^c) \leq  \mathbb{P}^n _\mu (\mathcal{A}^{\mathbf{m}_\mu/2}_{n}(\mu)) \leq 2^{k+1} \cdot e^{- \frac{n \mathbf{m}_\mu^2}{4}}, 
	\end{equation}
	the last from {\em Hoeffding's inequality} considering $k=\left| A_\mu\right| <\infty$.

	If we consider the events:
	\begin{align}\label{eq_proof_th_rate_data_driven_finite_support_4}
		\mathcal{C}^\epsilon_n(\mu) &=\left\{ x^n_1 \in \mathbb{X}^n:  \left|H(\hat{\mu}_{n,\epsilon_n})  - H({\mu}) \right|  >\epsilon \right\} \text{ and }\\
		\mathcal{D}^\epsilon_n(\mu) &=\left\{ x^n_1 \in \mathbb{X}^n: \left|H(\hat{\mu}_{n})  - H({\mu}) \right|  >\epsilon \right\}
	\end{align}
	and we use the fact that by definition $\hat{\mu}_{n,\epsilon_n}=\hat{\mu}_{n}$ conditioning on  $\mathcal{B}^{\epsilon_n}_n$,  it follows that $\mathcal{C}^\epsilon_n(\mu) \cap \mathcal{B}^{\epsilon_n}_n \subset \mathcal{D}^\epsilon_n(\mu)$.  	
	Then,  for all $\epsilon>0$ and $n\geq N$
	\begin{align}\label{eq_proof_th_rate_data_driven_finite_support_5}
	\mathbb{P}^n_\mu(\mathcal{C}^\epsilon_n(\mu)) &\leq \mathbb{P}^n_\mu(\mathcal{C}^\epsilon_n(\mu) \cap \mathcal{B}^{\epsilon_n}_n) + \mathbb{P}^n _\mu ((\mathcal{B}^{\epsilon_n}_n)^c)\nonumber\\
		&\leq \mathbb{P}^n_\mu(\mathcal{D}^\epsilon_n(\mu)) + \mathbb{P}^n _\mu ((\mathcal{B}^{\epsilon_n}_n)^c)\nonumber\\
		&\leq 2^{k+1}  \left[ e^{-\frac{2 n \epsilon ^2}{(\mathbf{M}_{\mu} + \frac{\log e}{\mathbf{m}_{\mu}})^2}} + e^{- \frac{n \mathbf{m}_\mu^2}{4}}  \right], 
	\end{align}
	the last inequality from Theorem \ref{th_con_inq} and (\ref{eq_proof_th_rate_data_driven_finite_support_3}).
	\end{proof}

\section{Summary and Final Remarks}
\label{sec_final}
In this work we show that entropy convergence results  
are instrumental to derive new (strongly consistent) estimation results for the Shannon entropy in infinite alphabets, and as a byproduct,  distribution estimators that are strongly consistent in direct and reverse  I-divergence.  
Adopting a set of sufficient conditions for entropy convergences in the context of four plug-in histogram-based schemes, we show concrete conditions where strong consistency for entropy estimation in $\infty$-alphabets can be obtained (Theorems \ref{pro_barron_conv}, \ref{th_barron} and \ref{th_data_driven}). In addition,  we explore the relevant  case where the target distribution has a finite but unknown support,  deriving for the overall estimation error 
almost sure rate of convergence results (Theorems \ref{th_con_inq} and \ref{th_rate_data_driven_finite_support}) that match the optimal asymptotic rate that can be obtained in the finite alphabet version of the problem. Finally, we focus on the case of a data-driven plug-in estimate that restricts the support where we estimate the distribution. The basic idea was to have design parameters to control estimation error effects in the context of entropy estimation. The conjecture was that this approach will offers degrees of freedom to restrict the learning complexity while introducing an approximation error and, consequently, the possibility to find an adequate balance between these two learning components.  Adopting a entropy convergence result (Lemma \ref{th_con_infty_1}), we show that this data-driven scheme offers the same universal estimation attributes than the classical plug-in estimate under some mild conditions on its threshold design parameter (Theorem \ref{th_data_driven}). In addition,  by addressing the  technical task of deriving concrete closed-form expressions for the estimation 
and approximation error in this adaptive data-driven context (where a subset of the space is selected adaptively from the data to restrict the estimation of the distribution), we show a solution were  almost sure rate of convergence of the overall estimation error are obtained over a family of distribution with some concrete tail bounded conditions (Theorems \ref{th_rate_data_driven_power} and \ref{th_rate_data_driven_exp_family}). These results show the potential that  data-driven frameworks has to adapt to the complexity of entropy estimation in unbounded alphabets by restricting the inference to a finite but dynamic subset of the space that scales with the amount of data and the distribution of the data in the sampling space. 

Concerning the classical plug-in estimate presented in Section \ref{sub_sec_fs_class_hist}, it is important to mention that the  work of  \cite{antos_2002}  shows that $\lim_{n \rightarrow \infty}  H(\hat{\mu}_n )=H(\mu)$ almost surely  distribution-free and,  furthermore, it provides rates of convergences for families with specific  
tail-bounded conditions \citep[Th. 7]{antos_2002}. On this context, Theorem \ref{th_con_inq}  focuses on the case when $\mu \in \mathcal{F}(\mathbb{X})$, where  new finite-length deviation inequalities and confidence intervals were derived.  From that perspective,  it complements the result presented in \citep{antos_2002} in the non-explored scenario when $\mu \in \mathcal{F}(\mathbb{X})$. It is also important to mention two results by \cite[Ths. 11 and 12]{ho_2010} for the plug-in estimate in (\ref{eq_sec_stat_learn_2}). They derive bounds for the object $\mathbb{P}_{\mu}^n(\left|H(\hat{\mu}_n)  - H({\mu}) \right|\geq \epsilon)$ 
and from this determine confidence intervals under a finite and known support restriction on the distribution $\mu$.
In contrast, our result in Theorem \ref{th_con_inq}  resolves the case for a  finite and unknown supported distribution,   that was declared to be a challenging problem based on the arguments presented in \cite[Th.13]{ho_2010} concerning the discontinuity of the entropy.

\acks{The work is supported by funding from FONDECYT Grant 1170854, CONICYT-Chile
and the Advanced Center for Electrical and Electronic Engineering (AC3E), Basal Project FB0008.  The author is grateful to Patricio Parada for his insights and stimulating discussion in the initial stage of this work.}

\appendix
\section{Minimax risk for finite Entropy distributions in Infinite alphabets}
\label{pro_unbounded_minimax_risk}
\begin{proposition} $R^*_n =\infty.$
\end{proposition}
For the proof, we use the following lemma that follows from \citep[Th. 1]{ho_2010}. 
\begin{lemma} \label{lemma_discontinity_entropy} 
Let us fix two arbitrary real numbers $\delta>0$ and $\epsilon>0$. Then there are $P$, $Q$ two finite supported distributions on $\mathbb{H}(\mathbb{X})$  that satisfy that $D(P||Q)<\epsilon$ while $H(Q)-H(P)>\delta$. \footnote{The proof of Lemma \ref{lemma_discontinity_entropy} derives from the same construction presented in the proof of \citep[Th. 1]{ho_2010}, i.e., $P=(p_1,..,p_L)$ and a modification of it $Q_M=(p_1\cdot (1-1/\sqrt{M}), p_2+p_1/M\sqrt{M},..., p_L+p_1/M\sqrt{M},p_1/M\sqrt{M},...,p_1/M\sqrt{M})$ both of finite support and consequently in $\mathbb{H}(\mathbb{X})$. It is simple to verify that as $M$ goes to infinity $D(P||Q_M) \longrightarrow 0$ while $H(Q_M)-H(P)\longrightarrow \infty$.} 
\end{lemma}

\begin{proof}
	For any pair of distribution $P$, $Q$ in $\mathbb{H}(\mathbb{X})$, {\em Le Cam's two point method} \citep{tsybakov_2009} shows that:   
	\begin{equation}\label{eq_pro_unbounded_minimax_risk_1}
		R^*_n \geq \frac{1}{4} \left( H(Q)-H(P) \right)^2 \exp^{-n D(P||Q)}. 
	\end{equation}	
	 Adopting Lemma \ref{lemma_discontinity_entropy} and (\ref{eq_pro_unbounded_minimax_risk_1}), for any $n$ and any arbitrary $\epsilon>0$ and $\delta>0$, we have that $R^*_n> \delta^2\exp^{-n\epsilon}/4$. Then exploiting the discontinuity of the entropy in infinite alphabets, we can fix $\epsilon$ and make $\delta$ arbitrar large. 
\end{proof}

\section{Proposition \ref{pro_est_error_barron}}
\label{appendix_3}
\begin{proposition}\label{pro_est_error_barron} 
	Under the assumptions of Theorem \ref{th_barron}: 
	\begin{align} \label{eq_proof_th_barron_9b}
	\lim_{n \rightarrow \infty} \sup_{x\in A_{\tilde{\mu}_n}}\left|\frac{d \tilde{\mu}_n}{d \mu^*_n}(x)-1\right|=0, \  \mathbb{P}_\mu-a.s.
	\end{align}
\end{proposition}
\begin{proof}
First note that $supp(\tilde{\mu}_n)=supp({\mu}^*_n)$, then $\frac{d \tilde{\mu}_n}{d \mu^*_n}(x)$ is well-defined and $\forall x \in A_{\tilde{\mu}_n}$
\begin{equation}\label{eq_proof_pro_est_error_barron_1}
	\frac{d \tilde{\mu}_n}{d \mu^*_n}(x)= \frac{(1-a_n)\cdot \mu(A_n(x)) +a_n v(A_n(x))}{(1-a_n)\cdot \hat{\mu}_n(A_n(x)) +a_n v(A_n(x))}.
\end{equation}
Then by construction, 
\begin{align}\label{eq_proof_pro_est_error_barron_2}
	\sup_{x\in A_{\tilde{\mu}_n}}\left|\frac{d \tilde{\mu}_n}{d \mu^*_n}(x)-1\right| \leq \sup_{A\in \pi_n}\frac{\left|\hat{\mu}_n(A)-\mu(A)\right|}{a_n\cdot h_n}.
\end{align}
From {\em Hoeffding's inequality} we have that  $\forall \epsilon>0$
\begin{align}\label{eq_proof_pro_est_error_barron_3}
\mathbb{P}^n_\mu \left( \sup_{A \in \pi_n} \left|\hat{\mu}_n(A) - \mu(A)\right| > \epsilon \right)
\leq 2 \cdot \left|\pi_n\right| \cdot \exp ^{- {2n \epsilon^2}}.
\end{align}
By condition ii),  given that $(1/a_nh_n)$ is $o(n^\tau)$
for some $\tau \in (0,1/2)$, then there exists $\tau_o \in (0,1)$
such that, 
\begin{align*}
	\lim_{n \rightarrow \infty}\frac{1}{n^{\tau_o}} \ln \mathbb{P}^n_\mu \left(\sup_{x\in A_{\tilde{\mu}_n} }\left|\frac{d \tilde{\mu}_n}{d \mu^*_n}(x)-1\right| > \epsilon \right) \leq 
	\lim_{n \rightarrow \infty}\frac{1}{n^{\tau_o}} \ln (2 \left|\pi_n\right| ) - {2\cdot (n^{\frac{1-\tau_o}{2}} a_n h_n \epsilon)^2}= - \infty.
\end{align*}
This implies that $\mathbb{P}^n_\mu \left(\sup_{x\in A_{\tilde{\mu}_n}}\left|\frac{d \tilde{\mu}_n}{d \mu^*_n}(x)-1\right| > \epsilon \right)$ is eventually dominated by a constant time $(e^{-n^{\tau_o}})_{n\geq 1}$, 
which from the {\em Borel-Cantelli lemma} \citep{varadhan_2001} implies that 
\begin{align}\label{eq_proof_pro_est_error_barron_5}
\lim_{n \rightarrow \infty} \sup_{x\in A_{\tilde{\mu}_n}}\left|\frac{d \tilde{\mu}_n}{d \mu^*_n}(x)-1\right|=0, \  \mathbb{P}_\mu-a.s.
\end{align}
\end{proof}

\section{Proposition \ref{pro_div_bound_ddp}}
\label{appendix_div_bound_ddp}
\begin{proposition}\label{pro_div_bound_ddp} 
	\begin{align*}
	D(\mu_{\epsilon_n} || \hat{\mu}_{n,\epsilon_n}) \leq  \frac{2 \log \frac{e}{\epsilon_n}}{\mu(\Gamma_{\epsilon_n})} \cdot V(\mu/\sigma_{\epsilon_n}, \hat{\mu}_n/\sigma_{\epsilon_n})
	\end{align*}
\end{proposition}
\begin{proof}
	From definition, 
	\begin{align} \label{eq_pro_div_bound_ddp_1}
	D(\mu_{\epsilon_n} || \hat{\mu}_{n,\epsilon_n}) = \frac{1}{\mu(\Gamma_{\epsilon_n})}  \sum_{x \in \Gamma_{\epsilon_n}} f_{\mu}(x) \log \frac{f_{\mu}(x)}{f_{\hat{\mu}_n}(x)} + \log \frac{\hat{\mu}_n(\Gamma_{\epsilon_n}) }{\mu(\Gamma_{\epsilon_n})}. 
	\end{align}
	For the right term in the RHS of (\ref{eq_pro_div_bound_ddp_1}): 
	\begin{align} \label{eq_pro_div_bound_ddp_1b}
	\log \frac{\hat{\mu}_n(\Gamma_{\epsilon_n}) }{\mu(\Gamma_{\epsilon_n})}  \leq \frac{\log(e) }{\mu(\Gamma_{\epsilon_n})}  \left|  \hat{\mu}_n(\Gamma_{\epsilon_n})-\mu(\Gamma_{\epsilon_n}) \right|.
	\end{align}
	For the left term in the RHS of (\ref{eq_pro_div_bound_ddp_1}): 
	\begin{align} \label{eq_pro_div_bound_ddp_2}
	 	\left| \sum_{x \in \Gamma_{\epsilon_n}} f_{\mu}(x) \log \frac{f_{\mu}(x)}{f_{\hat{\mu}_n}(x)} \right| 
			&= \left| \sum_{\substack{ x \in \Gamma_{\epsilon_n} \\ f_\mu(x)\leq f_{\hat{\mu}_n}(x)}} f_{\mu}(x) \log \frac{f_{\mu}(x)}{f_{\hat{\mu}_n}(x)} + \sum_{\substack{ x \in \Gamma_{\epsilon_n}\\  f_\mu(x) > f_{\hat{\mu}_n}(x) \geq \epsilon_n}} f_{\mu}(x) \log \frac{f_{\mu}(x)}{f_{\hat{\mu}_n}(x)}  \right| \nonumber\\
			&\leq \sum_{\substack{ x \in \Gamma_{\epsilon_n} \\ f_\mu(x)\leq f_{\hat{\mu}_n}(x)}} f_{\mu}(x) \log \frac{f_{\hat{\mu}_n}(x) }{f_{\mu}(x)} + \sum_{\substack{ x \in \Gamma_{\epsilon_n}\\  f_\mu(x) > f_{\hat{\mu}_n}(x) \geq \epsilon_n}}  f_{\hat{\mu}_n}(x) \log \frac{f_{\mu}(x)}{f_{\hat{\mu}_n}(x)}\nonumber\\
			&+ \sum_{\substack{ x \in \Gamma_{\epsilon_n}\\  f_\mu(x) > f_{\hat{\mu}_n}(x) \geq \epsilon_n}}  ( f_{{\mu}}(x)-f_{\hat{\mu}_n}(x)) \cdot  \log \frac{f_{\mu}(x)}{f_{\hat{\mu}_n}(x)}\\
			\label{eq_pro_div_bound_ddp_2b}
			&\leq  \log e \left[ \sum_{\substack{ x \in \Gamma_{\epsilon_n} \\ f_\mu(x)\leq f_{\hat{\mu}_n}(x)}} (f_{\hat{\mu}_n}(x)-f_{\mu}(x)) + \sum_{\substack{ x \in \Gamma_{\epsilon_n}\\  f_\mu(x) > f_{\hat{\mu}_n}(x)}} (f_{\mu}(x)- f_{\hat{\mu}_n}(x))  \right] \nonumber\\
			&+ \log \frac{1}{\epsilon_n} \cdot \sum_{\substack{ x \in \Gamma_{\epsilon_n}\\  f_\mu(x) > f_{\hat{\mu}_n}(x) }}  ( f_{{\mu}}(x)-f_{\hat{\mu}_n}(x))\\
			\label{eq_pro_div_bound_ddp_2c}
			&\leq  (\log e+ \log \frac{1}{\epsilon_n}) \cdot \sum_{x\in \Gamma_{\epsilon_n}} 	 \left| f_{\mu}(x)- f_{\hat{\mu}_n}(x) \right|. 
	\end{align}
	The first inequality in (\ref{eq_pro_div_bound_ddp_2}) is by triangular inequality,  the second in (\ref{eq_pro_div_bound_ddp_2b}) is from the fact that $\ln x\leq x-1$ for $x>0$. 
	Finally from definition of the total  variational distance over $\sigma_{\epsilon_n}$ in (\ref{eq_data_driven_new5b})
	we have that 
	\begin{align} \label{eq_pro_div_bound_ddp_3}
		2 \cdot V(\mu/\sigma_{\epsilon_n}, \hat{\mu}_n/\sigma_{\epsilon_n})=  \sum_{x\in \Gamma_{\epsilon_n}}  \left| f_{\mu}(x)- f_{\hat{\mu}_n}(x) \right| + \left|  \hat{\mu}_n(\Gamma_{\epsilon_n})-\mu(\Gamma_{\epsilon_n}) \right|, 
	\end{align}
	 which concludes the argument from (\ref{eq_pro_div_bound_ddp_1}), (\ref{eq_pro_div_bound_ddp_1b}) and (\ref{eq_pro_div_bound_ddp_2}).
\end{proof}

\section{Proposition \ref{pre1_th_conv_conditional_prob}}
\label{appendix_1}
\begin{proposition}\label{pre1_th_conv_conditional_prob}
	Considering that $(k_n) \rightarrow \infty$, 
	there exists $K>0$ and $N>0$ such that $\forall n \geq N$, 
	\begin{equation}\label{eq_proof_th_conv_conditional_prob_1}
		V(\tilde{\mu}_{k_n}, \hat{\mu}^*_{k_n,n} )  \leq K \cdot V(\mu/\sigma_{k_n}, \hat{\mu}_n/\sigma_{k_n}).
	\end{equation}
\end{proposition}
\begin{proof} 
\begin{align}\label{eq_proof_th_conv_conditional_prob_2}
	V(\tilde{\mu}_{k_n}, \hat{\mu}^*_{k_n,n} )&= \frac{1}{2} \sum_{x \in A_{\mu}\cap \Gamma_{k_n}} \left|\frac{\mu\left\{x\right\}}{\mu(\Gamma_{k_n})}- \frac{\hat{\mu}_n\left\{x\right\}}{\hat{\mu}_n(\Gamma_{k_n})}  \nonumber\right|\\
	&\leq \frac{1}{2\mu(\Gamma_{k_n})} \left[\sum_{x \in A_{\mu}\cap \Gamma_{k_n}}  \left|\hat{\mu}_n\left\{x\right\}- \mu\left\{x\right\}\right| + \sum_{x \in A_{\mu}\cap \Gamma_{k_n}} \hat{\mu}_n\left\{x\right\} \left| \frac{\mu(\Gamma_{k_n})}{\hat{\mu}_n(\Gamma_{k_n})}-1\right| \right] \nonumber\\
	&= \frac{1}{2\mu(\Gamma_{k_n})} \left[ 2 \cdot V(\mu/\sigma_{k_n}, \hat{\mu}_n/\sigma_{k_n}) + \left| \mu(\Gamma_{k_n})  - \hat{\mu}_n(\Gamma_{k_n}) \right| \right]\nonumber\\
	&\leq \frac{3 \cdot V(\mu/\sigma_{k_n}, \hat{\mu}_n/\sigma_{k_n})}{2\mu(\Gamma_{k_n})}
\end{align}
By the hypothesis  $\mu(\Gamma_{k_n}) \rightarrow 1$, which concludes the proof.
\end{proof}

\section{Proposition \ref{prop_oracle_appro_err_apr}}
\label{app_prop_oracle_appro_err_apr}
\begin{proposition}\label{prop_oracle_appro_err_apr}
	If $\epsilon_n$ is $O(n^{-\tau})$  with $\tau\in (0, 1/2)$, then 
	\begin{equation*}
	\lim \sup_{n \rightarrow \infty} b_{\epsilon_n}(X_1,..,X_n) - a_{2 \epsilon_n} \leq 0, \  \mathbb{P}_{\mu}-a.s.
	\end{equation*}
\end{proposition}
\begin{proof}
	Let us define the set 
	  $$\mathcal{B}_n=  \left\{(x_1,..,x_n):  \tilde{\Gamma}_{2 \epsilon_n}  \subset \Gamma_{\epsilon_n} \right\}\subset \mathbb{X}^n.$$  
	  From definition  every sequence $(x_1,..,x_n) \in \mathcal{B}_n$  is such that   
	  $b_{\epsilon_n}(x_1,..,x_n) \leq a_{2\epsilon_n}$ and, consequently,  we just need to prove that 
	   $\mathbb{P}_{\mu}(\lim \inf_{n \rightarrow \infty} \mathcal{B}_n) = \mathbb{P}_{\mu} (\cup_{n\geq 1}\cap_{k\geq n} \mathcal{B}_k)=1$ \citep{breiman_1968}. 
	   Furthermore, if $\sup_{x\in \tilde{\Gamma}_{2\epsilon_n} }  \left| \hat{\mu}_n(\left\{x \right\})- \mu(\left\{x \right\}) \right| \leq \epsilon_n$, then by definition of $\tilde{\Gamma}_{2 \epsilon_n}$ in (\ref{pr_th_rate_data_driven_1}),  we have that $\hat{\mu}_n (\left\{ x\right\})  \geq \epsilon_n$ for all $x\in \Gamma_{2 \epsilon_n}$ (i.e., $\tilde{\Gamma}_{2 \epsilon_n} \subset \Gamma_{\epsilon_n}$).  From this 
	   \begin{align}\label{pr_th_rate_data_driven_4}   
	   	\mathbb{P}^n_{\mu} (\mathcal{B}^c_n)  &\leq  \mathbb{P}^n_{\mu}  \left( \sup_{x\in \tilde{\Gamma}_{2\epsilon_n}}  \left| \hat{\mu}_n(\left\{x \right\})- \mu(\left\{x \right\}) \right| > \epsilon_n \right)  \leq  \left| \tilde{\Gamma}_{2\epsilon_n} \right| \cdot e^{-2n \epsilon_n^2} \leq   
		 \frac{1}{2\epsilon_n} \cdot e^{-2n \epsilon_n^2} ,
	   \end{align} 
	   from the {\em Hoeffding's inequality} \citep{devroye1996,devroye_2001}, the union bound and the fact that 
	   by construction $ \left| \tilde{\Gamma}_{2 \epsilon_n} \right| 
	   \leq \frac{1}{2\epsilon_{n}}$.
	   If we consider $\epsilon_{n}=O(n^{-\tau})$ and $l>0$, we have that: 
	   \begin{align}\label{pr_th_rate_data_driven_4b}   
	   	\frac{1}{n^l} \cdot \ln \mathbb{P}^n_{\mu} (\mathcal{B}^c_n) \leq \frac{1}{n^l} \ln (1/2 \cdot n^{\tau}) - 2n^{1-2\tau -l}.    
	   \end{align}
	   From (\ref{pr_th_rate_data_driven_4b}) for any  $\tau \in (0,1/2)$ there is $l \in (0,1-2\tau]$
	   such that $\mathbb{P}^n_{\mu} (\mathcal{B}^c_n)$ is bounded by a term $O(e^{-n^l})$. 
	   This implies that $\sum_{n\geq 1}\mathbb{P}^n_{\mu} (\mathcal{B}^c_n) <\infty$, that suffices to  show that  $\mathbb{P}_{\mu} (\cup_{n\geq 1}\cap_{k\geq n} \mathcal{B}_k)=1$. 
\end{proof}

\section{Auxiliary results for Theorem \ref{th_rate_data_driven_power}}
\label{aux_th_rate_data_driven}
Let us first consider the series 
\begin{align}\label{eq_aux_th_rate_data_driven_1}
	\mathcal{S}_{x_o}= \sum_{x\geq x_o} x^{-p}&=x_o^{-p} \cdot  \left( 1+  \left( \frac{x_o}{x_o+1}\right)^p+ \left(\frac{x_o}{x_o+2}\right)^p + \ldots  \right) \nonumber\\
					&=x_o^{-p} \cdot  \left( \tilde{\mathcal{S}}_{x_o,0}+  \tilde{\mathcal{S}}_{x_o,1} + \ldots  + \tilde{\mathcal{S}}_{x_o,x_o-1} \right),  
\end{align}
where $ \tilde{\mathcal{S}}_{x_o,j} \equiv \sum_{k=0}^{\infty}  \left( \frac{k\cdot x_o +j}{x_o}\right)^{-p} $ for all 
$j\in \left\{0,..,x_o-1 \right\}$. It is simple to verify that for all 
$j\in \left\{0,..,x_o-1 \right\}$, $\tilde{\mathcal{S}}_{x_o,j} \leq \tilde{\mathcal{S}}_{x_o,0} =\sum_{k\geq 0} k^{-p} <\infty$ given that  by hypothesis $p>1$. Consequently, 
$\mathcal{S}_{x_o}\leq x_o^{1-p}\cdot \sum_{k\geq 0} k^{-p}$.

Similarly, for the second series we have that:
\begin{align}\label{eq_aux_th_rate_data_driven_2}
	 \mathcal{R}_{x_o} = \sum_{x\geq x_o}x^{-p} \log x &= x_o^{-p} \cdot   \left(\log (x_o)+  \left(\frac{x_o}{x_o+1} \right) \log (x_o+1) +  \left(\frac{x_o}{x_o+2} \right) \log (x_o+2) + \ldots \right)  \nonumber\\
	 &=x_o^{-p} \cdot \left( \tilde{\mathcal{R}}_{x_o,0}+  \tilde{\mathcal{R}}_{x_o,2} + \ldots  +\tilde{\mathcal{R}}_{x_o,x_o-1} \right), 
\end{align}
where $ \tilde{\mathcal{R}}_{x_o,j} \equiv \sum_{k=1}^{\infty}  \left( \frac{k\cdot x_o +j}{x_o} \right)^{-p} \cdot \log (kx_o+j)$  for all  $j\in \left\{0,..,x_o-1 \right\}$. 
Note again that  $\tilde{\mathcal{R}}_{x_o,j}\leq   \tilde{\mathcal{R}}_{x_o,0} <\infty$ for all $j\in \left\{0,..,x_o-1 \right\}$, and, consequently, $\mathcal{R}_{x_o} \leq x_o^{1-p} \cdot \sum_{k\geq 1} k^{-p} \log k$ from (\ref{eq_aux_th_rate_data_driven_2}).

\section{Proposition \ref{prop_oracle_est_err_aprt}}
\label{app_prop_oracle_est_err_aprt}
\begin{proposition}\label{prop_oracle_est_err_aprt}
	If $\epsilon_n$ is $O(n^{-\tau})$  with $\tau\in (0, 1/2)$, then 
	\begin{equation*}
	\lim \inf_{n \rightarrow \infty} \xi_n(X_1,..,X_n) - \rho_n(X_1,..,X_n) \geq 0, \ \ \mathbb{P}_\mu-a.s.
	\end{equation*}
\end{proposition}
\begin{proof}
	By definition if $\sigma{(\Gamma_{\epsilon_n})} \subset \sigma( \tilde{\Gamma}_{\epsilon_n/2})$ then 
	$\xi_n(X_1,..,X_n) \geq \rho_n(X_1,..,X_n)$. Consequently,  if we define the set: 
	\begin{equation} \label{eq_app_prop_oracle_est_err_aprt_1}
		\mathcal{B}_n= \left\{ (x_1,..,x_n): \sigma{(\Gamma_{\epsilon_n})} \subset \sigma( \tilde{\Gamma}_{\epsilon_n/2})\right\}, 
	\end{equation}
	then the proof reduced to verify that $\mathbb{P}_\mu(\lim\inf_{n  \rightarrow \infty} \mathcal{B}_n)$ $= \mathbb{P}_\mu( \cup_{n\geq 1 } \cap_{k\geq n} \mathcal{B}_k)=1$.
	
	On the other hand, if $\sup_{x\in \Gamma_{\epsilon_n}}  \left|  \hat{\mu}_n(\left\{ x\right\})  -  \mu ( \left\{ x \right\} ) \right| \leq \epsilon_n/2$ then by definition of $\Gamma_\epsilon$, for all $x\in \Gamma_{\epsilon_n}$ $\mu (\left\{ x\right\}) \geq \epsilon_n/2$, i.e., $\Gamma_{\epsilon_n} \subset  \tilde{\Gamma}_{\epsilon_n/2}$.  In other words, 
	\begin{equation} \label{eq_app_prop_oracle_est_err_aprt_2}
	\mathcal{C}_n= \left\{ (x_1,..,x_n): \sup_{x\in \Gamma_{\epsilon_n}}  \left|  \hat{\mu}_n(\left\{ x\right\})  -  \mu ( \left\{ x \right\} ) \right| \leq \epsilon_n/2 \right\} \subset \mathcal{B}_n.
	\end{equation}
	Finally, 
	\begin{align} \label{eq_app_prop_oracle_est_err_aprt_3}
		\mathbb{P}^n_{\mu} (\mathcal{C}^c_n)  & = \mathbb{P}^n_{\mu}  \left( \sup_{x\in \Gamma_{\epsilon_n}}  \left|  \hat{\mu}_n(\left\{ x\right\})  -  \mu ( \left\{ x \right\} ) \right| > \epsilon_n/2 \right) \leq \left| \Gamma_{\epsilon_n} \right| \cdot  e^{-n \epsilon^2/2} \leq  \frac{1}{\epsilon_n} \cdot  e^{-n \epsilon^2/2}.
	\end{align}
	In this context, if we consider $\epsilon_{n}=O(n^{-\tau})$ and $l>0$, then we have that: 
	   \begin{align}\label{eq_app_prop_oracle_est_err_aprt_4}   
	   	\frac{1}{n^l} \cdot \ln \mathbb{P}^n_{\mu} (\mathcal{C}^c_n) \leq \tau \cdot \frac{\ln n}{n^l} - \frac{n^{1-2\tau -l}}{2}.    
	   \end{align}
	   Therefore, we have that for any  $\tau \in (0,1/2)$ we can take $l \in (0,1-2\tau]$
	   such that $\mathbb{P}^n_{\mu} (\mathcal{C}^c_n)$ is bounded by a term $O(e^{-n^l})$.  Then, 
	  the Borel Cantelli lemma tells us that 
	   $\mathbb{P}_\mu( \cup_{n\geq 1 } \cap_{k\geq n} \mathcal{C}_k)=1$, which concludes the proof from (\ref{eq_app_prop_oracle_est_err_aprt_2}).
\end{proof}

\section{Proposition \ref{prop_rate_oracle_est_error}}
\label{app_rate_oracle_est_error}
\begin{proposition} \label{prop_rate_oracle_est_error}
For the $p$-power tail dominating distribution stated in Theorem \ref{th_rate_data_driven_power}, 
if $(\epsilon_n)$ is $O(n^{-\tau})$ with $\tau\in (0,p)$ then
$\xi_n(X_1,..,X_n) \text{ is } o(n^{-q}) \text{ for all } q\in (0,(1-\tau/p)/2)$, $\mathbb{P}_\mu$-a.s.
\end{proposition}
\begin{proof}
	From the {\em Hoeffding's inequality} we have that
	\begin{align} \label{eq_prop_rate_oracle_est_error_1}
		\mathbb{P}^n_{\mu} \left(  \left\{ x_1,..,x_n:   \xi_n(x_1,..,x_n) > \delta  \right\}  \right) &\leq \left| \sigma(\tilde{\Gamma}_{\epsilon_n/2}) \right| \cdot  e^{-2n \frac{\delta^2}{\log(1/\epsilon_n)^2}} \nonumber\\
					&\leq 2^{(\frac{2k_o}{\epsilon_n})^{1/p}+1} \cdot e^{-2n \frac{\delta^2}{\log(1/\epsilon_n)^2}},
	\end{align}
	the second inequality using that $\tilde{\Gamma}_\epsilon \leq (\frac{k_0}{\epsilon}) ^{1/p}+1$  from the definition 
	of $\tilde{\Gamma}_{\epsilon}$ in (\ref{pr_th_rate_data_driven_1}) and the tail bounded assumption on $\mu$.
	If we consider $\epsilon_{n}=O(n^{-\tau})$ and $l>0$, then we have that: 
	\begin{align} \label{eq_prop_rate_oracle_est_error_2}
		\frac{1}{n^l} \cdot \ln \mathbb{P}^n_{\mu} \left(  \left\{ x_1,..,x_n:   \xi_n(x_1,..,x_n) > \delta  \right\}  \right) 
			& \leq \ln 2  \cdot (C n^{\tau/p-l}+n^{-l}) -  \frac{2\delta^2}{\tau^2} \cdot \frac{n^{1-l}}{\log n^2}
	\end{align}
	for some constant $C>0$. Then in order to obtain that $\xi_n(X_1,..,X_n)$ converges almost surely to zero from (\ref{eq_prop_rate_oracle_est_error_2}),  it is sufficient that $l>0$, $l<1$, and $l>\tau/p$. This implies that if $\tau<p$, 
	there is $l\in (\tau/p,1)$ such that such that $\mathbb{P}^n_{\mu} (\xi_n(x_1,..,x_n) > \delta )$ is bounded 
	by a term $O(e^{-n^l})$ and, consequently,  $\lim_{n \rightarrow \infty} \xi_n(X_1,..,X_n)=0$, $\mathbb{P}_\mu$-a.s.\footnote{Using the same steps used in Appendix \ref{app_prop_oracle_est_err_aprt}.}
	
	Moving to the rate of convergence of $ \xi_n(X_1,..,X_n)$ (assuming that $\tau<p$), 
	let us consider $\delta_n=n^{-q}$ for some $q\geq 0$.  From (\ref{eq_prop_rate_oracle_est_error_2}): 
	\begin{align} \label{eq_prop_rate_oracle_est_error_3}
		\frac{1}{n^l} \cdot \ln \mathbb{P}^n_{\mu} \left(  \left\{ x_1,..,x_n:   \xi_n(x_1,..,x_n) > \delta_n  \right\}  \right) 
			& \leq \ln 2  \cdot (C n^{\tau/p-l}+n^{-l}) -  \frac{2\delta^2}{\tau^2} \cdot \frac{n^{1-2q-l}}{\log n^2}.
	\end{align}
	To make $\xi_n(X_1,..,X_n)$ being $o(n^{-q})$ $\mathbb{P}$-a.s., a sufficient condition is that 
	$l>0$, $l>\tau/p$, and $l<1-2q$. Therefore (considering that $\tau<p$),  the admissibility condition on the existence of a exponential rate of convergence $O(e^{-n^l})$ for $l>0$ for the deviation event $\left\{ x_1,..,x_n:   \xi_n(x_1,..,x_n) > \delta_n  \right\}$ is that $\tau/p < 1-2q$, which is equivalent to $0<q<\frac{1-\tau/p}{2}$.
\end{proof}

\bibliography{main_jorge_silva}			

\begin{thebibliography}{36}
\providecommand{\natexlab}[1]{#1}
\providecommand{\url}[1]{\texttt{#1}}
\expandafter\ifx\csname urlstyle\endcsname\relax
  \providecommand{\doi}[1]{doi: #1}\else
  \providecommand{\doi}{doi: \begingroup \urlstyle{rm}\Url}\fi

\bibitem[Antos and Kontoyiannis(2001)]{antos_2002}
A.~Antos and I.~Kontoyiannis.
\newblock Convergence properties of fucntionals estimates for discrete
  distributions.
\newblock \emph{Random Structures and Algorithms}, 19\penalty0 (3-4):\penalty0
  163--193, October-December 2001.

\bibitem[Barron et~al.(1992)Barron, Gy{\"o}rfi, and {van der
  Meulen}]{barron_1992}
A.~Barron, L.~Gy{\"o}rfi, and E.~C. {van der Meulen}.
\newblock Distribution estimation consistent in total variation and in two
  types of information divergence.
\newblock \emph{IEEE Transactions on Information Theory}, 38\penalty0
  (5):\penalty0 1437--1454, 1992.

\bibitem[Beirlant et~al.(1997)Beirlant, Dudewicz, Gy{\"o}rfi, and {van der
  Meulen}]{beirlant_1997}
J.~Beirlant, E.~Dudewicz, L.~Gy{\"o}rfi, and E.~C. {van der Meulen}.
\newblock Nonparametric entropy estimation: An overview.
\newblock \emph{Int. Math. Statist. Sci.}, 6:\penalty0 17--39, 1997.

\bibitem[Berlinet et~al.(1998)Berlinet, Vajda, and {van der
  Meulen}]{berlinet_1998}
A.~Berlinet, I.~Vajda, and E.~C. {van der Meulen}.
\newblock About the asymptotic accuracy of {B}arron density estimate.
\newblock \emph{IEEE Transactions on Information Theory}, 44\penalty0
  (3):\penalty0 999--1009, 1998.

\bibitem[Breiman(1968)]{breiman_1968}
Leo Breiman.
\newblock \emph{Probability}.
\newblock Addison-Wesley, 1968.

\bibitem[Cover and Thomas(2006)]{cover_2006}
T.~M. Cover and J.~A. Thomas.
\newblock \emph{Elements of Information Theory}.
\newblock Wiley Interscience, New York, 2do edition, 2006.

\bibitem[Csisz\'{a}r(1967)]{csiszar_1967}
I.~Csisz\'{a}r.
\newblock Information-type measures of difference of probability distributions
  and indirect observations.
\newblock \emph{Studia Scient. Mathe. Hung.}, 2:\penalty0 299--318, 1967.

\bibitem[Csisz\'{a}r and Shields(2004)]{csiszar_2004}
I.~Csisz\'{a}r and P.~C. Shields.
\newblock \emph{Information theory and Statistics: A tutorial}.
\newblock Now Inc., 2004.

\bibitem[Darbellay and Vajda(1999)]{darbellay_1999}
Georges~A. Darbellay and Igor Vajda.
\newblock Estimation of the information by an adaptive partition of the
  observation space.
\newblock \emph{IEEE Transactions on Information Theory}, 45\penalty0
  (4):\penalty0 1315--1321, 1999.

\bibitem[Devroye and Lugosi(2001)]{devroye_2001}
L.~Devroye and G.~Lugosi.
\newblock \emph{Combinatorial Methods in Density Estimation}.
\newblock Springer - Verlag, New York, 2001.

\bibitem[Devroye et~al.(1996)Devroye, Gy{\"o}rfi, and Lugosi]{devroye1996}
L.~Devroye, L.~Gy{\"o}rfi, and G.~Lugosi.
\newblock \emph{A Probabilistic Theory of Pattern Recognition}.
\newblock New York: Springer-Verlag, 1996.

\bibitem[Gray(1990)]{gray_1990_b}
R.~M. Gray.
\newblock \emph{Entropy and Information Theory}.
\newblock Springer - Verlag, New York, 1990.

\bibitem[Gy{\"o}rfi et~al.(1994)Gy{\"o}rfi, P\'{a}li, and {van der
  Meulen}]{gyorfi_1994b}
L.~Gy{\"o}rfi, I.~P\'{a}li, and E.~C. {van der Meulen}.
\newblock There is no universal source code for an infinity source alphabet.
\newblock \emph{IEEE Transactions on Information Theory}, 40\penalty0
  (1):\penalty0 267--271, 1994.

\bibitem[Herremo\"{e}s(2007)]{harremoes_2007}
P.~Herremo\"{e}s.
\newblock \emph{Entropy. Search, Complexity}, volume~16 of \emph{Bolyai Society
  of Mathematical Studies}, chapter Information topologies with applications,
  pages 113--150.
\newblock Sprimger, 2007.

\bibitem[Ho and Yeung(2009)]{ho_2009}
S.-W. Ho and R.~W. Yeung.
\newblock On the discontinuity of the {S}hannon information measures.
\newblock \emph{IEEE Transactions on Information Theory}, 55\penalty0
  (12):\penalty0 5362--5374, December 2009.

\bibitem[Ho and Yeung(2010)]{ho_2010}
S.-W. Ho and R.~W. Yeung.
\newblock The interplay between entropy and variational distance.
\newblock \emph{IEEE Transactions on Information Theory}, 56\penalty0
  (12):\penalty0 5906--5929, December 2010.

\bibitem[Jiao et~al.(2015)Jiao, Venkat, Han, and Weissman]{jiao_2015}
J.~Jiao, K.~Venkat, Y.~Han, and T.~Weissman.
\newblock Minimax estimation of functionals of discrete distributions.
\newblock \emph{IEEE Transactions on Information Theory}, 61\penalty0 (5),
  2015.

\bibitem[Kemperman(1969)]{kemperman_1969}
J.~Kemperman.
\newblock On the optimum rate of transmitting information.
\newblock \emph{Ann. Math. Statist.}, 40:\penalty0 2156--2177, December 1969.

\bibitem[Kullback(1967)]{kullback_1967}
S.~Kullback.
\newblock A lower bound for discrimination in terms of variation.
\newblock \emph{IEEE Transactions on Information Theory}, 13:\penalty0
  126--127, 1967.

\bibitem[Kullback and Leibler(1951)]{kullback_1951}
S.~Kullback and R.~Leibler.
\newblock On information and sufficiency.
\newblock \emph{Ann. Math. Statist.}, 22:\penalty0 79--86, 1951.

\bibitem[Lugosi and Nobel(1996)]{lugosi_1996_ann_sta}
G.~Lugosi and A.~B. Nobel.
\newblock Consistency of data-driven histogram methods for density estimation
  and classification.
\newblock \emph{The Annals of Statistics}, 24\penalty0 (2):\penalty0 687--706,
  1996.

\bibitem[Nobel(1996)]{nobel_1996b}
Andrew~B. Nobel.
\newblock Histogram regression estimation using data-dependent partitions.
\newblock \emph{The Annals of Statistics}, 24\penalty0 (3):\penalty0
  1084--1105, 1996.

\bibitem[Paninski(2004)]{paniski_2004}
L.~Paninski.
\newblock Estimating entropy on m bins given fewer than m samples.
\newblock \emph{IEEE Transactions on Information Theory}, 50\penalty0
  (9):\penalty0 2200--2203, 2004.

\bibitem[Piera and Parada(2009)]{piera_2009}
F.~Piera and P.~Parada.
\newblock On convergence properties of {S}hannon entropy.
\newblock \emph{Problems of Information Transmission}, 45\penalty0
  (2):\penalty0 75--94, 2009.

\bibitem[Rissanen(2010)]{rissanen_2010}
Jorma Rissanen.
\newblock \emph{Information and Complexity in Statistical Modeling}.
\newblock Sprimger, 2010.

\bibitem[Silva and Narayanan(2010{\natexlab{a}})]{silva_2010}
J.F. Silva and S.~N. Narayanan.
\newblock Non-product data-dependent partitions for mutual information
  estimation: Strong consistency and applications.
\newblock \emph{IEEE Transactions on Signal Processing}, 58\penalty0
  (7):\penalty0 3497--3511, July 2010{\natexlab{a}}.

\bibitem[Silva and Narayanan(2010{\natexlab{b}})]{silva_2010b}
J.F. Silva and S.N. Narayanan.
\newblock Information divergence estimation based on data-dependent partitions.
\newblock \emph{Journal of Statistical Planning and Inference}, 140\penalty0
  (11):\penalty0 3180 -- 3198, November 2010{\natexlab{b}}.

\bibitem[Silva and Narayanan(2012)]{silva_2012}
J.F. Silva and S.N. Narayanan.
\newblock Complexity-regularized tree-structured partition for mutual
  information estimation.
\newblock \emph{IEEE Transactions on Information Theory}, 58\penalty0
  (3):\penalty0 940 -- 1952, March 2012.

\bibitem[Silva and Parada(2012)]{silva_isit_2012}
J.F. Silva and P.~Parada.
\newblock Shannon entropy convergence results in the countable infinite case.
\newblock In \emph{International Symposium on Information Theory}. IEEE, June
  2012.

\bibitem[Tsybakov(2009)]{tsybakov_2009}
A.~B. Tsybakov.
\newblock \emph{Introduction to non-parametric estimation}.
\newblock New York: Springer-Verlag, 2009.

\bibitem[Vajda and {van der Meulen}(2001)]{vajda_2001}
I.~Vajda and E.~C. {van der Meulen}.
\newblock Optimization of {B}arron density estimates.
\newblock \emph{IEEE Transactions on Information Theory}, 47\penalty0
  (5):\penalty0 1867--1883, July 2001.

\bibitem[Valiant and Valiant(2010)]{valiant_2010}
G.~Valiant and P.~Valiant.
\newblock A {CLT} and tight lower bounds for estimating entropy.
\newblock In \emph{Proc. Electrom. Colloq. Comput. Complex.}, volume~17, page
  179, 2010.

\bibitem[Valiant and Valiant(2011)]{valiant_2011}
G.~Valiant and P.~Valiant.
\newblock Estimating the unseen: An $n/log(n)$-sample estimator for entropy and
  support size, shown opitmal via new clts.
\newblock In \emph{in Proc. 43rd Annu. ACM Symp. Theory Comput.}, pages
  685--694, 2011.

\bibitem[{Van der Vaart}(2000)]{vaart_2000}
A.~W. {Van der Vaart}.
\newblock \emph{Asymtotic Statistics}, volume~3.
\newblock Cambridge Univ Press, 2000.

\bibitem[Varadhan(2001)]{varadhan_2001}
S.R.S. Varadhan.
\newblock \emph{Probability Theory}.
\newblock American Mathematical Society, 2001.

\bibitem[Wu and Yang(2016)]{wu_2016}
Y.~Wu and P.~Yang.
\newblock Minimax rates of entropy estimation on large alphabets via best
  plynomial approximation.
\newblock \emph{IEEE Transactions on Information Theory}, 62\penalty0 (6),
  2016.

\end{thebibliography}

\end{document}